\documentclass[a4paper,11pt]{article}

\usepackage[utf8]{inputenc}

\usepackage{thm-restate}

\usepackage{amssymb}
\usepackage{amsmath,thmtools}
\usepackage{mathtools}
\usepackage{physics}

\usepackage[margin=2.54cm]{geometry}

\usepackage{tabularx}
\usepackage{caption}    
\usepackage{float}
\usepackage{subfig}
\usepackage{multirow}
\usepackage{wrapfig}

\usepackage{cite}

\usepackage{subfiles}

\usepackage{algorithm}

\usepackage{appendix}

\usepackage[export]{adjustbox}
\usepackage{xcolor}
\usepackage{qcircuit}
\usepackage{dsfont}
\usepackage{enumitem}
\usepackage{authblk}
\usepackage{bm}

\usepackage[pagebackref,unicode]{hyperref}
\usepackage[nameinlink,capitalize]{cleveref}

\hypersetup{
    colorlinks=true, 
    linkcolor=blue, 
    citecolor=blue, 
    urlcolor=blue 
}

\usepackage{amsthm}
\usepackage{thmtools,thm-restate}
\newtheorem{theorem}{Theorem}

\newtheorem{lemma}{Lemma}
\newtheorem{corollary}{Corollary}
\newtheorem{definition}{Definition}
\newtheorem{proposition}{Proposition}

\makeatletter
\newtheorem*{rep@theorem}{\rep@title}
\newcommand{\newreptheorem}[2]{%
\newenvironment{rep#1}[1]{%
 \def\rep@title{#2 \ref{##1}}%
 \begin{rep@theorem}}%
 {\end{rep@theorem}}}
\makeatother

\newtheorem{subdefinition}{Definition}[definition]

\newreptheorem{theorem}{Theorem}
\newreptheorem{lemma}{Lemma}

\theoremstyle{definition}

\newtheorem{remark}{Remark}

\newcommand{\mO}{\mathcal{O}}

\newcommand{\bmO}{\bm{\mathcal{O}}}

\newcommand{\Pclass}{\mathsf{P}}
\newcommand{\PSPACE}{\mathsf{PSPACE}}
\newcommand{\PH}{\mathsf{PH}}
\newcommand{\NP}{\mathsf{NP}}

\newcommand{\MA}{\mathsf{MA}}
\newcommand{\AM}{\mathsf{AM}}

\newcommand{\QMA}{\mathsf{QMA}}

\newcommand{\QCMA}{\mathsf{QCMA}}
\newcommand{\QCAM}{\mathsf{QCAM}}
\newcommand{\QCIP}{\mathsf{QCIP}}
\newcommand{\IP}{\mathsf{IP}}
\newcommand{\QIP}{\mathsf{QIP}}

\newcommand{\PCP}{\mathsf{PCP}}
\newcommand{\QPCP}{\mathsf{QPCP}}
\newcommand{\QCPCP}{\mathsf{QCPCP}}

\newcommand{\BQP}{\mathsf{BQP}}
\newcommand{\BPP}{\mathsf{BPP}}

\newcommand{\BQ}{\mathsf{BQ}}
\newcommand{\BP}{\mathsf{BP}}

\newcommand{\PromiseNP}{\mathsf{PromiseNP}}

\newcommand{\poly}{\mathrm{poly}}
\newcommand{\polylog}{\mathrm{polylog}}

\newcounter{custalgocounter}
\renewcommand{\thecustalgocounter}{\arabic{custalgocounter}}

\usepackage{mdframed}

{%
    \end{mdframed}
  \end{minipage}
}

\crefname{conjecture}{Conjecture}{conjectures}
\Crefname{conjecture}{Conjecture}{Conjectures}

 \title{\vspace{-2.4cm}Classical versus quantum queries in quantum PCPs\\ with classical proofs}
\author[1,2]{Harry Buhrman}
\author[3]{Fran{\c c}ois Le Gall}
\author[4]{Jordi Weggemans}
\affil[1]{Quantinuum, United Kingdom}
\affil[2]{QuSoft \& University of Amsterdam, Amsterdam, The Netherlands}
\affil[3]{Nagoya University, Japan}
\affil[4]{QuSoft \& CWI, Amsterdam, the Netherlands}

\begin{document}

\maketitle

\begin{abstract}
We generalize quantum-classical PCPs, first introduced by Weggemans, Folkertsma and Cade (TQC 2024), to allow for $q$ \textit{quantum} queries to a polynomially-sized classical proof ($\QCPCP_{Q,c,s}[q]$). Exploiting a connection with the polynomial method, we prove that for any constant $q$, promise gap $c-s = \Omega(1/\poly(n))$ and $\delta>0$, we have $\QCPCP_{Q,c,s}[q] \subseteq \QCPCP_{1-\delta,1/2+\delta}[3] \subseteq \BQ \cdot \NP$, where $\BQ \cdot \NP$ is the class of promise problems with quantum reductions to an $\NP$-complete problem. Surprisingly, this shows that we can amplify the promise gap from inverse polynomial to constant for constant query quantum-classical PCPs, and that any quantum-classical PCP making any constant number of quantum queries can be simulated by one that makes only three classical queries. Nevertheless, even though we can achieve promise gap amplification, our result also gives strong evidence that there exists no constant query quantum-classical PCP for $\QCMA$, as it is unlikely that $\QCMA \subseteq \BQ \cdot \NP$, which we support by giving oracular evidence. In the (poly-)logarithmic query regime, we show for any positive integer $c$, there exists an oracle relative to which $\QCPCP[\mO(\log^c n)] \subsetneq \QCPCP_Q[\mO(\log^c n)]$, contrasting the constant query case where the equivalence of both query models holds relative to any oracle. Finally, we connect our results to more general quantum-classical interactive proof systems.
\end{abstract}

\setcounter{tocdepth}{2}
{\small \tableofcontents}

\section{Introduction}
A probabilistically checkable proof (PCP) system consists of a polynomial-time verifier that uses $r(n)$ random coins and has query access to some proof provided by a prover, to which it can make at most $q(n)$ queries. The celebrated PCP theorem states that the full power of $\NP$ can be captured by a PCP that only has $q(n) = \mO(1)$ and $r = \mO(\log n)$~\cite{Arora1998proof,Arora1998probabilistic,Dinur2007thepcp}. In quantum complexity theory, one of the biggest open questions is whether the \textit{de facto} quantum generalization of $\NP$, i.e.~$\QMA$, can also be characterized in terms of a \textit{quantum} probabilistically checkable proof system~\cite{aharonov2013guest}. This question is formalized as the quantum PCP conjecture, and is usually considered in terms of its local Hamiltonian formulation: this states that it is $\QMA$-hard to estimate the ground state energy of a local Hamiltonian up to constant additive error relative to the operator norm. Much less studied is its proof-checking formulation, which posits that any promise problem in $\QMA$ can be verified by a quantum polynomial-time verifier that only acts non-trivially on a subset of a polynomially-sized proof (effectively tracing out the rest). Both formulations are known to be equivalent under \emph{quantum} reductions~\cite{Aharonov2008The,aharonov2013guest,buhrman2024quantum}.

Ref.~\cite{weggemans2023guidable} proposed to study an intermediate version of a quantum probabilistically checkable proof system, where the verifier remains quantum but the proof is classical ($\QCPCP[q]$). They showed that for a constant amount of classical queries, the corresponding class is contained in $\BQP^{\NP[1]}$, which is $\BQP$ allowed to make a single query to an $\NP$ oracle. Whilst the definition of $\QCPCP[q]$ preserves the \textit{locality} aspect of PCPs, it does not capture the \textit{query} aspect as usually considered in a quantum setting: one generally considers \textit{quantum} queries instead of classical access to a string. This means that in a single query, the whole proof can be ``accessed'' in superposition, a model which is known to have an exponential advantage over classical queries for some computational tasks.\footnote{See \url{https://quantumalgorithmzoo.org} for a list of examples. Generally, query lower bounds in the randomized setting also hold in the quantum setting when only classical queries are allowed, provided the lower bound method works even when all intermediate computations (in between the queries) can be inefficient. This is because any quantum query algorithm in this model can be simulated with exponential overhead (in time and space) by a classical computation, but using the same number of queries.}
Hence, the following question arises:
\begin{center}    
\noindent \textit{Can quantum-classical PCPs be made more powerful when they are allowed to make quantum queries to a classical proof?}
\end{center}

\subsection{Our results}
In this work, we generalize the class $\QCPCP_{c,s}[q]$ to $\QCPCP_{Q,c,s}[q]$ such that it captures the standard quantum query model to an unknown bit string~\cite{buhrman2002complexity}. A $\QCPCP_{Q,c,s}[q]$-verifier consists of a uniform family of quantum circuits $V = \{V_n : n \in \mathbb{N}\}$, where the description of each $V_n$ is generated by a fixed polynomial-time Turing machine given $1^n$ as an input. We have that each circuit $V_n$ is allowed to make $q$ queries to a classical proof $y \in \{0,1\}^{\poly(n)}$ provided by an untrusted prover, via a unitary $U_y$ such that $U_y \ket{i} \ket{a} = \ket{i} \ket{a \oplus y_i}$.\footnote{One can also view this as having the proof stored in a quantum read-only memory (QROM), and the verifier is only allowed to access the memory $q$ number of times.}  We then say that a promise problem $A = (A_{\text{\sc yes}}, A_{\text{\sc no}})$ is in $\QCPCP_{Q,c,s}[q]$ if for all inputs $x\in A_{\text{\sc yes}}$, there exists a proof $y$ such that the $\QCPCP_{Q,c,s}[q]$-verifier accepts with probability at least $c$, and for all $x\in A_{\text{\sc no}}$, the verifier accepts all proofs $y$ with probability at most $s$ (\cref{def:QCPCP} and~\ref{def:QCPCP_Q2}). We will study quantum-classical PCPs in both types of query models in the setting where the number of queries is constant or (poly)-logarithmic.

\paragraph{Constant number of proof queries.}
We show that in the constant query regime the following result holds.
\begin{theorem}[From~\cref{cor:containment_BQNP} and~\cref{prop:QCPCP_equiv}]
For any positive integer $q \in \mathbb{N}$ and for any computable functions $c,s$ such that $c-s \geq 1/\poly(n)$, we have 
\begin{align*}
    \QCPCP_{Q,c,s}[q] \subseteq \QCPCP_{1-\delta,1/2+\delta}[3] \subseteq \BQ \cdot \NP,
\end{align*}
for any constant $\delta >0$ .
\label{thm:1_ext}
\end{theorem}

\noindent Here `$\BQ$' is a quantum extension of Sch\"oning's $\BP$-operator~\cite{schoning1989probabilistic},\footnote{Both are examples of what is known as dot operators~\cite{borchert2001dot}.} such that $\BQ \cdot \mathcal{C}$ for a class $\mathcal{C}$ contains all promise problems that have a quantum reduction to a promise problem $A$ that is complete for $\mathcal{C}$. Hence, when $\mathcal{C}$ is a \emph{classical} complexity class, one can view the $\BQ$-operator as ``pulling the quantumness'' out of a problem. In~\cref{sec:BQ}, we prove that the $\BQ$-operator satisfies many of the similar properties as the $\BP$-operator (\cref{prop:prop_BQ}) and also discuss its alternative formulation (for classes $\mathcal{C} \supseteq \Pclass$) in terms of feeding random strings generated by measuring the output states of quantum algorithms as extra input to the class (\hyperref[def:BQ_alt]{Definition~\ref*{def:BQ_alt}}).

The following conclusions can be drawn from~\Cref{thm:1_ext}:
\begin{enumerate}[label=(\roman*)]
    \item Any constant query quantum-classical PCP protocol can be simulated by a quantum-classical PCP making only $3$ classical queries which has a constant promise gap. Even more surprisingly, this even holds when the original completeness/soundness gap was inverse polynomial(!) instead of constant, showing that amplification in this regime can be done without increasing the query count.\footnote{If the original number of queries was at least $3$.}
    \item It states that one \textit{can} pull out the quantumness of quantum-classical PCPs, in terms of its interpretation in the context of the $\BQ$-operator, no matter if the access to the proof is quantum or classical. Since it seems very unlikely that $\QCMA \subseteq \BQ \cdot \NP$---as it would imply that the ``quantum part'' in the computation does not have to use the proof---this provides very strong evidence that it is unlikely that there exists some notion of a ``quantum PCP'' that uses a classical proof.  
\end{enumerate}

We strengthen point (ii) by giving a classical oracle relative to which $\QCPCP_Q[\mO(1)]$ does not capture the power of $\QCMA$.

\begin{restatable}{theorem}{orsepQCMA}
For any positive integer $q \in \mathbb{N}$, there exists an oracle $\bmO$ relative to which
\begin{align*}
     \QCPCP[q]^{\bmO} \subseteq \QCPCP_Q[q]^{\bmO} \subsetneq  \QCMA^{\bmO}.
\end{align*}
\label{thm:2_ext}
\end{restatable}

\noindent Unconditionally proving $\QCPCP_Q[\mO(1)] \neq \QCMA$ would imply $\Pclass \neq \PSPACE$ and therefore seems to be beyond the current techniques.

We believe that Point (i) highlights the intuition that a classical proof in a quantum verification setting can be viewed as being ``uncompiled.'' In the fully quantum case, it's easy to show that there exists a polynomial $p(n)$ such that $\QPCP_{c,s}[2] = \QMA$ when $c-s = 1/p(n)$; roughly speaking, one picks a term of a $\QMA$-hard $2$-local Hamiltonian uniformly at random, and applies a $3$-qubit unitary operation which maps an ancillary qubit in $\ket{0}$ to $\ket{1}$ proportional to the energy of the $2$-local density matrix with respect the sampled local term.\footnote{See for example~\cite{Kitaev2002ClassicalAQ} and~\cite{buhrman2024quantum}. Similarly, in the fully classical case, one can sample a term of a $\NP$-hard $3$-local constraint satisfaction problem uniformly at random and detect unsatisfiability with a probability $1/m$, where $m$ is the total number of constraints.} Therefore, the whole game of proving quantum PCP conjecture is to show that for some $q\in \mO(1)$, we have $\QPCP_{c,s}[2] = \QPCP_{c',s'} [q]$ with $c'-s' = \Omega(1)$.  For quantum PCPs with classical proofs (independent of the access model),~\cref{thm:1_ext} shows that such an amplification step is indeed possible, however, we lose the property that ``local information'' about a proof is enough to verify any problem in $\QCMA$ with an inverse polynomially small promise gap. In particular, variants of the local Hamiltonian which are $\QCMA$-complete use a classical proof to describe a quantum circuit which prepares a quantum state with low energy (or has a large overlap with such a state)~\cite{Wocjan2003two,weggemans2023guidable}, which can again only be ``locally verified'' after the quantum state is prepared, i.e., compiled.

\paragraph{(Poly-)logarithmic number of queries.} 
We show that when the number of queries to the proof becomes (poly-)logarithmic both query models can be separated relative to an oracle.

\begin{restatable}{theorem}{orseplog}  
For any positive integer $c \in \mathbb{N}$, there exists an oracle $\bmO$ such that 
\begin{align*}
\QCPCP[\mO(\log^c n)]^{\bmO} \varsubsetneq \QCPCP_Q [\mO(\log^c n)]^{\bmO}.
\end{align*}
\label{thm:3_ext}
\end{restatable}

\noindent This contrasts the constant query case, where~\cref{thm:1_ext} holds even in the presence of oracles.

\subsection{Proof ideas}
To prove~\cref{thm:1_ext}, we exploit the connection between quantum-classical PCPs and the \textit{polynomial method} of~\cite{beals2001quantum}. The polynomial method is a technique used to prove lower bounds on quantum query complexity, and utilizes the idea that the square of the amplitudes of every $q$-query quantum algorithm to a string $y \in \{0,1\}^N$ can be expressed as a real-valued multi-linear polynomial $P(y)$ of degree $2q$ on $N$ bits. In our case, we will be interested in the case when $N = |y| = \poly(n)$ and $q$ is constant, so $P$ consists of at most $\binom{|y|}{2q} = \poly(n)$ number of monomials. We prove that for a fixed input $x \in \{0,1\}^n$, an approximation of $P$ can be efficiently learned by sampling from the measurement outcome of the runs $\QCPCP_Q[q]$-verification circuit (\cref{thm:reduction}). In these runs, it never queries the actual proof through $U_y$, instead making queries to constructed $U_{y^S}$'s for some ``fake'' proofs $y^{S}$, which can be efficiently implemented given a description of $y^{S}$. Here $S \subseteq [N]$ indicates a subset of at most $2q$ variables participating in a term. Combined with the promise on the acceptance probabilities of the quantum-classical PCP verifier in the {\sc yes}- and {\sc no}-cases, this yields a \emph{quantum reduction} to a \emph{multi-linear polynomial threshold problem}: here you are given classical descriptions of the coefficients ${\beta_S}$ up to a certain number of bits of precision of some constant degree multi-linear polynomial $P : \{0,1\}^{p(n)} \rightarrow \mathbb{R}$ with $p(n) = \poly(n)$, and an efficiently computable number $a \in [0,1]$, and the task is to decide whether there exists a $y \in \{0,1\}^{p(n)}$ such that $P(y) \geq a$ or for all $y$ it holds that $P(y) < a$. The multi-linear polynomial threshold problem is clearly in $\NP$, as the smallest possible function value of $P$ above $a$ can be discriminated from $a$ using only a polynomial number of bits, which follows from the fact that each of the coefficients is specified up to a certain number of bits. Using a classical PCP construction, one then also has that the multi-linear polynomial threshold problem is in $\PCP[3,\mO(\log r)]$~\cite{HastadSome1997}. Moreover, conditioning on the quantum reduction succeeding, the output of the reduction can be made deterministic, which means that the prover can fix the proof used by the classical PCP verification protocol in the {\sc yes}-case.

The above construction is somewhat similar in spirit to a recent work by Arad and Santha~\cite{arad2024quasi}, in which a ``quasi-quantum PCP theorem'' in terms of a local Hamiltonian problem over so-called quasi-quantum states is shown by using a classical PCP construction to achieve amplification. However, a key difference is that they want to reduce \emph{to} a (quasi-)quantum problem from an amplified CSP, whilst we reduce \emph{from} a quantum problem to a CSP, which can then be gap amplified.

The proof of~\cref{thm:2_ext} relies on the OR $\circ$ Forrelation oracle from~\cite{aaronson2021acrobatics}, which was used to demonstrate an oracle separation between $\NP^{\BQP}$ and $\BQP^{\NP}$. We extend this result to show that the same oracle separation holds when $\NP^{\BQP}$ is replaced by $\QCMA$. Additionally, we make use of the fact that the inclusion $\QCPCP_Q[\mO(1)] \subseteq \BQ \cdot \NP$ from~\cref{thm:1_ext} holds even under relativization.

The key insight in proving~\cref{thm:3_ext} is to consider the decision version of a search problem, specifically computing the OR-function over a string of poly-logarithmic length. Using quantum queries, the Bernstein-Vazirani algorithm~\cite{bernstein1993quantum} can decode $\mO(\log n)$ bits with just a single quantum query to a string of polynomial size. By concatenating $\mO(\log^c n)$ of such strings into a proof, we can learn a total of $\mO(\log^{c+1} n)$ bits using only $\mO(\log^c n)$ quantum queries, after which a single classical query is sufficient to solve the decision problem. To show that this is not achievable with classical access to a proof, we establish a quantum lower bound on computing the OR function for $n$ bits, even in cases where $k$ bits (from a potentially larger proof) are observed to aid in finding the hidden string. Our lower bound generalizes a result from~\cite[Appendix D]{buhrman2024quantum}. Although their proof technique could also be applied to achieve the same result, we argue that our approach is simpler and more straightforward.

\subsection{Implication to quantum-classical interactive proof systems}
As a final side result, our oracle separation between $\QCMA$ and $\BQ \cdot \NP$ also leads to a corollary regarding quantum-classical interactive proof systems. We define a new class of quantum-classical interactive proof systems, denoted $\QCIP[k]$, which, to the best of our knowledge, has not been explored in the literature. This class contains all promise problems that can be decided by a $k$-message interaction between a polynomial-time quantum verifier and a prover, where only classical strings are exchanged between them.

\begin{restatable}[Quantum-classical interactive proofs]{definition}{QCIPdef}
Let $n \in \mathbb{N}$ and $p(n) : \mathbb{N} \rightarrow \mathbb{N}$ be a polynomial. A promise problem $A = (A_{\textup{\sc yes}}, A_{\textup{\sc no}})$ is in $\QCIP_{c,s}[k]$ if there exists a $\Pclass$-uniform family of polynomial-time quantum verifier circuits $\{V_n : n \in \mathbb{N}\}$, where each verifier $V_n$ takes $x \in \{0,1\}^n$ as an input and exchanges $k$ classical messages of length at most $p(n)$ with a computationally unbounded prover $P : \{0,1\}^* \rightarrow \{0,1\}^*$, such that the following conditions hold:
\begin{itemize}
    \item If $x \in A_{\textup{\sc yes}}$, there exists a prover $P$ that causes $V_n$ to accept with probability at least $c$.
    \item If $x \in A_{\textup{\sc no}}$, then for every prover $P$, $V_n$ accepts with probability at most $s$.
\end{itemize}
We define $\QCIP[k] = \QCIP_{2/3, 1/3}[k]$, $\QCIP = \bigcup_{\alpha \geq 1} \QCIP[n^\alpha]$, and $\QCMA = \QCIP[1]$.
\label{def:QCIP}
\end{restatable}

We have that $\QCIP$ in itself is not a particularly interesting class, as it is already known that $\QIP = \IP = \PSPACE$. Therefore, since $\PSPACE = \mathsf{IP} \subseteq \QCIP$~\cite{shamir1992ip} and $\QCIP \subseteq \mathsf{QIP} = \PSPACE$~\cite{jain2011qip}, we actually have $\QCIP = \QIP = \IP = \PSPACE$. This implies that, for interactive proof systems allowing a polynomial number of interaction rounds, the computational power is independent of whether the proofs are classical or quantum, or whether the verifier is quantum. However, it is well-known that $\QIP[3] = \QIP$~\cite{kitaev2000parallelization}, which demonstrates that quantum interactive proofs with a constant number of rounds (specifically, $k \geq 3$) are more powerful than classical ones, as $\IP[k] = \AM$ for all $k \geq 2$. It seems unlikely that a similar result holds for quantum-classical interactive proofs, as the proof that $\QIP[3] = \QIP$ heavily relies on the ability to exchange quantum registers, creating entanglement between the verifier and the prover’s quantum states.

Trivially, for any $k\geq 1$ and any oracle $\bmO$, we have that $\QCMA^{\bmO} \subseteq \QCIP[k]^{\bmO}$ (the verifier does not have to make use of their ability to send a message to the prover, yielding the exact same class). Using the oracle $\bmO$ from~\cref{def:oracle} and the result from~\cref{prop:OracleSep}, we have that $ \left(\BQ \cdot \NP\right)^{\bmO} \subseteq \BQP^{\NP^{\bmO}}  \not\supset \QCMA^{\bmO} \subseteq \QCIP[k]^{\bmO}$, leading directly to the following corollary.
\begin{restatable}{corollary}{corQCIP}
For any $k\geq 1$, there exists an oracle $\bmO$ relative to which
\begin{align*}
    \QCIP[k]^{\bmO} \not\subset \left(\BQ \cdot \NP \right)^{\bmO}.
\end{align*}
\label{cor:QCIP_oracle}
\end{restatable}

The same result holds by considering $\QCAM[k]$ instead of $\QCIP[k]$, where is $\QCAM[k]$ the quantum generalization of $\AM[k]$ (see, for example,~\cite{agarwal2024quantum,aaronson2023certified}). This contrasts the classical result that $\AM[k] \subseteq \IP[k] =  \BP \cdot \NP$ for all $k \in \mathbb{N}_{\geq 2}$, which holds relative to all oracles. 

Intuitively, we believe that our statements for quantum-classical PCPs hold because one \emph{can} pull out the quantumness, while~\cref{cor:QCIP_oracle} for quantum-classical interactive proof systems holds because one \emph{cannot} pull out the quantumness.

\subsection{Discussion and open questions}
In this work, we have introduced a quantum version of the $\BP$-operator and applied it to explore interactive proof systems with a quantum verifier and classical messages. For quantum-classical PCPs, many open questions appear to be resolved by this work, as achieving an unconditional separation from $\QCMA$ would require overcoming the $\Pclass$ versus $\PSPACE$ barrier, which seems beyond current techniques.

For future work, we would like to extend the direction of this work to general quantum-
classical interactive proof systems, for which we believe that the concept of not being able to ``pull
out the quantumness'' has many implications. We conjecture that, relative to an oracle, one can demonstrate that no round reduction exists for a constant number of rounds, and that $\QCAM[k]$ does not contain $\QCIP[k]$, highlighting the distinction between ``quantumness'' and randomness. Another direction worth exploring is whether, for some fixed $k \geq 2$, $\QCIP[k]$ admits perfect completeness, possibly with a constant overhead in the number of rounds. Additionally, what is the relationship between $\QCIP[k]$ and $\QIP[2]$ for larger values of $k$?

For the $\BQ$-operator, it would be interesting to investigate whether it has broader applications. Specifically, are there other problems that allow a quantum reduction to a classical problem? Some examples in a non-quantum PCP context can be found in cryptography~\cite{regev2009lattices,debris2023quantum}, although the reduction in~\cite{regev2009lattices} has also been shown to exist in a classical setting~\cite{Peikert2009Public}.

\section{Preliminaries}
\label{sec:prelim}

\paragraph{Notation} 
We use $[N]$ to denote the set $\{1, \dots, N\}$. For a random variable $X$, $\Pr[X]$ and $\mathbb{E}[X]$ represent its probability and expectation value, respectively. Let $G_T = \{0, \frac{1}{T}, \frac{2}{T}, \dots, 1 - \frac{1}{T}, 1\}$ be the set of $T+1$ evenly spaced numbers in the interval $[0, 1]$. We denote the $d$-dimensional identity matrix as $\mathbb{I}_d$, and omit the subscript when the dimension is clear from the context. Finally, $\norm{\cdot}_{\infty}$ refers to the uniform norm, i.e., for any real- or complex-valued bounded function $f$ defined on a set $S$, the uniform norm is given by $\norm{f}_{\infty} = \sup_{x \in S} \abs{f(x)}$.

\subsection{Complexity theory}
\paragraph{Promise versus classes of decision problems.} Almost all classes in this work will be defined with respect to \textit{promise} problems: a promise problem $A$ is defined by a tuple $A = (A_\textup{\sc yes}, A_\textup{\sc no}, A_\textup{\sc inv})$, where $A_\textup{\sc yes} \cap A_\textup{\sc no} \cap A_\textup{\sc inv} = \emptyset$ and $A_\textup{\sc yes} \cup A_\textup{\sc no} \cup A_\textup{\sc inv} = \{0,1\}^*$. We say $A_\textup{\sc yes}$, $A_\textup{\sc no}$ and $A_\textup{\sc inv}$ are  the set of {\sc yes}-instances, {\sc no}-instances, and {\sc Invalid}-instances, respectively. We will usually denote a promise problem as $A = (A_\textup{\sc yes}, A_\textup{\sc no})$, as the sets $A_\textup{\sc yes}$ and $A_\textup{\sc no}$ implicitly specify $A_\textup{\sc inv}$. On the contrary, a \textit{decision} problem $D$ (or a \emph{language}) consists of only two non-intersecting sets $(D_\textup{\sc yes}, D_\textup{\sc no})$ with $D_\textup{\sc yes} \cup D_\textup{\sc no} = \{0,1\}^*$. This means that $D$ can alternatively be viewed as a promise problem $D' = (D'_\textup{\sc yes}, D'_\textup{\sc no}, D'_\textup{\sc inv})$ where $D'_\textup{\sc yes} = D_\textup{\sc yes}$, $D'_\textup{\sc no} = D_\textup{\sc no}$, and $D'_\textup{\sc inv} = \emptyset$ (i.e., the promise holds on all instances $x \in \{0,1\}^*$). Another way to view it is to note that decision problems are characterized by \textit{total} Boolean functions, while promise problems are characterized by \textit{partial} Boolean functions. Given a (promise) problem $A = (A_\textup{\sc yes}, A_\textup{\sc no})$, we say that an algorithm $\mathcal{A}$ \textit{accepts} when it outputs `$x \in A_\textup{\sc yes}$', and it \textit{rejects} when it outputs `$x \in A_\textup{\sc no}$'. For classical algorithms, this occurs when the algorithm terminates with the most significant output bit being one. For a quantum algorithm, this occurs when a measurement in the computational basis of the final state yields a measurement outcome of `1' for the first qubit. If $\mathcal{A}$ is run on an instance $x \in A_\textup{\sc inv}$, it is allowed to output arbitrarily.

\paragraph{Complexity classes.} We assume basic familiarity with complexity classes; see the Complexity Zoo for precise definitions.\footnote{\url{https://complexityzoo.net/Complexity_Zoo}.} In this work, all \textit{quantum} classes will be understood as promise classes, while classical deterministic classes will be treated as such only if explicitly stated. For example, when we write $\BQP$, we implicitly mean $\mathsf{PromiseBQP}$, but we will explicitly write $\PromiseNP$ when referring to $\NP$ where the verifier is allowed to output arbitrarily on inputs outside the promise.

\paragraph{Quantum query model.} In quantum query algorithms, there are typically two types of query access to an oracle that encodes a Boolean function $ f: \{0, 1\}^n \rightarrow \{0, 1\} $. In the standard query model, the oracle acts as an $n+1$-qubit unitary operator $ U_f $, defined by its action on basis states as
\begin{align*}
U_f \ket{x} \ket{a} = \ket{x} \ket{a \oplus f(x)}
\end{align*}
where $a \in \{0,1\}$. Another model is the \textit{phase query}, where the oracle is accessed via a unitary operator $ U_{f,\textup{phase}}$ defined by
\begin{align*}
U_{f,\textup{phase}} \ket{x} = (-1)^{f(x)} \, \ket{x}.
\end{align*}
Using the phase kickback trick, a phase query can be implemented at the cost of a single standard query. Therefore, in this work, we will always assume standard query access and freely use both phase and standard queries without further justification. For a promise class $\mathcal{C}$, we denote $V^{\mathcal{C}}$ to indicate that an algorithm $V$ has access to an oracle for any problem $A = (A_\text{{\sc yes}}, A_\text{{\sc no}}, A_\text{{\sc inv}})$ in $\mathcal{C}$. If $V$ makes invalid queries (i.e., $x \in A_\text{{\sc inv}}$), the oracle may respond arbitrarily with a {\sc yes} (`$1$') or {\sc no} (`$0$') answer~\cite{goldreich2006promise, gharibian2019complexity}. In this work, the same principle will be applied when an oracle encodes a partial Boolean function $f$ and queries are made outside the domain of $f$.

\subsection{The OR \texorpdfstring{$\circ$}{circ} Forrelation oracle } 
\label{subsec:or_forr_oracle}
We will use an oracle used in~\cite{aaronson2021acrobatics}, which crucially relies on the following result by Raz and Tal~\cite{raz2022oracle}.

\begin{lemma}[From~\cite{raz2022oracle}, Theorem 1.2] For all sufficiently large $N$, there exists an explicit distribution $\mathcal{F}_N$ that we call the Forrelation distribution over $\{0,1\}^N$ such that:
\begin{enumerate}
    \item There exists a quantum algorithm $\mathcal{A}$ that makes $\polylog(N)$ queries and runs in time $\polylog(N)$ such that
    \begin{align*}
        \abs{\Pr_{x \in \mathcal{F}_N} [\mathcal{A}(x)=1] - \Pr_{y \{0,1\}^N} [\mathcal{A}(y)=1] } \geq 1-\frac{1}{N^2}. 
    \end{align*}
    \item For any $C \in \mathsf{AC}^0[\text{quasipoly}(N),\mO(1)]$:
    \begin{align*}
        \abs{\Pr_{x \in \mathcal{F}_N} [C(x)=1] - \Pr_{y \{0,1\}^N} [C(y)=1] } \leq \frac{\text{polylog}(N)}{\sqrt{N}}
    \end{align*}
\end{enumerate}
\label{lem:forr}
\end{lemma}
The precise definition of the Forrelation distribution is not relevant for our purposes, but can be found in~\cite{raz2022oracle} In~\cite{aaronson2021acrobatics} this result is used to provide oracle separations between $\BQP^\PH$ and $\PH^{\BQP}$. We will use the following specific oracle from\cite{aaronson2021acrobatics}.

\begin{definition}[OR $\circ$ Forrelation oracle, adapted from~\cite{aaronson2021acrobatics}] We define $\bmO$ as the oracle such that for each $n\in \mathbb{N}$, we add into $\bmO$ a region $\mathcal{R}$ consisting of a function $f_n : \{0,1\}^{n^2} \times \{0,1\}^{n^2} \rightarrow \{0,1\}$, defined as follows:
\begin{enumerate}[label=(\roman*)]
    \item If $0^n \notin L^{\bmO}$ (i.e. $L^{\bmO}(0^n)=0$), then we have that $f_n(x,y)$ is drawn uniformly at random for all $(x,y)\in \{0,1\}^{n^2}$ .
    \item If $0^n \in L^{\bmO}$ (i.e. $L^{\bmO}(0^n)=1$), then there exists one $\hat{x}$, drawn uniformly at random from the set $\{0,1\}^{n^2}$ such that for all $y\in \{0,1\}^{n^2}$ the function values of $f_n(\hat{x},y)$ are drawn from the Forrelation distribution $\mathcal{F}_{2^{n^2}}$. For all $x \in \{0,1\}^{n^2}\setminus \{\hat{x}\}$, we again have that all entries $y\in \{0,1\}^{n^2}$ of  $f_n(x,y)$ are drawn uniformly at random.
\end{enumerate}
Outside the region $\mathcal{R}$ the oracle $\bmO$ always outputs $0$.
\label{def:oracle}
\end{definition}
Let us show that the same oracle implies an oracle separation between $\QCMA$ and $\BQP^{\NP}$. We only have to show that the oracle problem is contained in $\QCMA^{\bmO}$, as non-containment in $\BQP^{\NP}$ is already shown in~\cite{aaronson2021acrobatics}.
\begin{proposition}\label{prop:OracleSep}
There exists an oracle $\bmO$ relative to which
\begin{align*}
    \QCMA^{\bmO} \not\subset  \BQP^{\NP^{\bmO}}.
\end{align*}
\end{proposition}

\begin{proof}
The proof follows the construction used in the proofs of Corollary 48 and Claim 49 in~\cite{aaronson2021acrobatics}. Let $L$ be a uniformly random unary language. Let $\mathcal{D}$ be the resulting distribution over oracles $\bmO$. We will now show that $L^{\bmO} \in \QCMA^{\bmO}$ with probability $1$ over $\bmO \sim \mathcal{D}$. Let $M^{\bmO}(0^n,z)$ be a $\QCMA^{\bmO}$ verifier which takes as input $0^n$ and witness $z \in \{0,1\}^{n^2}$, and runs the quantum algorithm $\mathcal{A}$ of~\cref{lem:forr} on $f_n(z,y)$. In the case when $0^n \in L$, this witness will be $z=\hat{x}$, and in the case $0^n \notin L$, it can be any $z \in \{0,1\}^{n^2}$. It follows from~\cref{lem:forr} that with $N=2^n$, we have
\begin{align*}
    \Pr_{\bmO \sim \mathcal{D}} [M^{\bmO}(0^n,z) \neq L^{\bmO}(0^n) ] \leq 2^{-2n^2}.
\end{align*}
By Markov's inequality, we then have
\begin{align*}
    \Pr_{\bmO \sim \mathcal{D}} [\Pr [M^{\bmO}(0^n,z) \neq L^{\bmO}(0^n)] \geq 1/3] \leq 3\cdot 2^{-2n^2}.
\end{align*}
Using the Borel-Cantelli Lemma, we can now argue that with probability $1$ over $\bmO$, we have that $M^{\bmO}$ correctly decides $L^{\bmO}(0^n)$ for all but finitely many $n \in \mathbb{N}$. We have
\begin{align*}
    \Pr_{\bmO \sim \mathcal{D}}[M^{\bmO} \text{ does not decide } L^{\bmO}(0^n)] \leq \sum_{n=1}^\infty 3\cdot 2^{-2n^2} < \infty.
\end{align*}
Hence, the probability that $M^{\bmO}$ fails on infinitely many inputs $n$ is zero. Consequently, we have that $M^{\bmO}$ can be modified into a $\QCMA^{\bmO}$ algorithm that correctly decides $L^{\bmO}(0^n)$ for all $n \in \mathbb{N}$, with probability $1$ over $\bmO \sim \mathcal{D}$. The proof that $L^{\bmO} \notin \BQP^{\PH^{\bmO}}$ (and therefore $L^{\bmO} \notin \BQP^{\NP^{\bmO}}$) is given in~\cite{aaronson2021acrobatics}, Corollary 48.
\end{proof}

\section{Quantum reductions and the \texorpdfstring{$\BQ$}{BQ}-operator}
\label{sec:BQ}
In this section, we will briefly review Sch\"oning's $\BP$-operator and some of its properties, and introduce a quantum analogue, which we denote as the `$\BQ$-operator'. We will see that it shares many of the same properties as the $\BP$-operator.

\subsection{Randomized reductions and Sch\"oning's \texorpdfstring{$\BP$}{BP}-operator}
Let us first recall the definition of randomized reductions.
\begin{definition}[Randomized reductions]
Let $A = (A_\textup{\sc yes},A_\textup{\sc no})$ and $B = (B_\textup{\sc yes},B_\textup{\sc no})$ be promise problems. We say that $A \leq_r B$ if there exists a polynomial-time probabilistic Turing machine $M$ such that:
\begin{itemize}
    \item Completeness: If $x \in A_\textup{\sc yes}$, then $\Pr_{z} [M(x,z) \in B_\textup{\sc yes}] \geq 2/3$,
    \item Soundness: If $x \in A_\textup{\sc no}$, then $\Pr_{z} [M(x,z) \in B_\textup{\sc no}] \geq 2/3$,
\end{itemize}
where $z$ are the random bits used by the Turing machine $M$.
\label{def:randred}
\end{definition}

\noindent Note that randomized reductions as defined here are not transitive: i.e., $A \leq_r B $ and $B \leq_r C$ do not necessarily imply $A \leq_r C$.\footnote{This property would hold if the probabilities in the randomized reductions could be amplified.} 

Using~\cref{def:randred}, the $\BP$-operator applied to a class with a complete problem $B$ gives a new class that contains all problems with randomized reductions to this complete problem $B$. 

\begin{definition}[$\BP$-operator]
Let $\mathcal{C}$ be a class of (promise) problems. Then $\BP \cdot \mathcal{C}$ consists of all (promise) problems $A = (A_\textup{\sc yes},A_\textup{\sc no})$ that can be reduced to a complete (promise) problem $B = (B_\textup{\sc yes},B_\textup{\sc no}) \in \mathcal{C}$ by a polynomial-time randomized reduction with success probability $\geq 2/3$, i.e.,
\begin{align*}
    \BP \cdot \mathcal{C} = \{A : A \leq_r B\}.
\end{align*}
\label{def:BP}
\end{definition}
The following alternative definition is also sometimes adopted in the literature.

\begin{subdefinition} 
Let $\mathcal{C}$ be any class of promise problems. Then $\BP \cdot \mathcal{C}$ consists of all promise problems $A = (A_\textup{\sc yes},A_\textup{\sc no})$ for which there exists a promise problem $B = (B_\textup{\sc yes},B_\textup{\sc no}) \in \mathcal{C}$, and a polynomial $p$, such that for all $x$ with $|x| = n$, the following conditions hold:
\begin{itemize}
    \item Completeness: If $x \in A_\textup{\sc yes}$, then $\Pr_{z} [(x,z) \in B_\textup{\sc yes}] \geq 2/3$,
    \item Soundness: If $x \in A_\textup{\sc no}$, then $\Pr_{z} [(x,z) \in B_\textup{\sc no}] \geq 2/3$,
\end{itemize}
where $z \in \{0,1\}^{p(n)}$ is uniformly distributed.
\label{def:BP_alt}
\end{subdefinition}

It is easy to see that both definitions are equivalent when $\mathcal{C} \supseteq \Pclass$, by taking $M(x,z)$ from~\cref{def:randred} to be the instance $(x,z)$ from~\hyperref[def:BP_alt]{Definition~\ref*{def:BP_alt}}. The condition $\mathcal{C} \supseteq \Pclass$ is necessary because~\cref{def:BP} always captures at least the power of $\BPP$, whereas this might not necessarily be true for~\hyperref[def:BP_alt]{Definition~\ref*{def:BP_alt}} when $\mathcal{C}$ is a very small class. From the above definitions, it is also clear that $\BP \cdot \Pclass = \BPP$. 

Let us now define some basic properties that are known for the $\BP$-operator. For this, we need the notion of \textit{majority reducibility}: we say a problem $A$ is majority reducible to a problem $B$, denoted $A \leq_\text{maj}^p B$, if there is a polynomial-time computable function $f$ mapping strings to sequences of strings such that for all $x$, if $f(x) = (y_1,\dots,y_k)$, then:
\begin{itemize}
    \item  $x \in A_\textup{\sc yes} \Rightarrow y_i \in B_\textup{\sc yes}$ for more than $k/2$ of the indices $i$,
    \item  $x \in A_\textup{\sc no} \Rightarrow $ $y_i \in B_\textup{\sc no}$ for more than $k/2$ of the indices $i$.
\end{itemize}

It is well-known that if a class $\mathcal{C}$ is closed under majority reducibility, i.e., if $A \leq_\text{maj}^p B$ and $B \in \mathcal{C}$ implies that $A \in \mathcal{C}$, then the soundness and completeness parameters of class $\BP \cdot \mathcal{C}$ can be amplified to become inverse exponentially close to $1$ and $0$, respectively~\cite{schoning1989probabilistic}.

\begin{proposition}[Some properties of the $\BP$-operator, adapted from~\cite{kobler2012graph}]
For all classes $\mathcal{C}, \mathcal{D}$, the following hold:
\begin{enumerate}[label=(\roman*)]
    \item $\mathcal{C} \subseteq \mathcal{D}$ implies $\BP \cdot \mathcal{C} \subseteq \BP \cdot \mathcal{D}$.
    \item $\BP \cdot \BP \cdot \mathcal{C} = \BP \cdot \mathcal{C}$ if the class $\mathcal{C}$ is closed under majority reducibility.
    \item $\BP \cdot \mathcal{C} \subseteq \BPP^{\mathcal{C}}$.
\end{enumerate}
\label{prop:prop_BP}
\end{proposition}

Properties (i) and (iii) follow directly from the definition of the $\BP$-operator, while (ii) can be shown using probability amplification~\cite{kobler2012graph}. If we take $\mathcal{C} = \NP$, and for example, choose $\mathsf{3}\text{-}\mathsf{SAT}$ as the $\NP$-complete problem (the result is independent of the choice of the $\NP$-complete problem), we arrive at $\BP \cdot \NP$, i.e., the class of all problems that have randomized reductions to $\mathsf{3}\text{-}\mathsf{SAT}$. It is a well-known result that $\BP \cdot \NP = \AM$~\cite{schoning1989probabilistic}. Finally, note that the $\BP$-operator can indeed be viewed as an operator that `pulls out the randomness' from a complexity class. For example, we have $\BP \cdot \MA = \BP \cdot \NP$, since $\BP \cdot \MA \subseteq \BP \cdot \AM = \BP \cdot \BP \cdot \NP = \BP \cdot \NP$ and $\BP \cdot \NP \subseteq \BP \cdot \MA$, using properties (i) and (ii) from~\cref{prop:prop_BP}.

\subsection{The \texorpdfstring{$\BQ$}{BQ}-operator}
Now that we have set the stage by recalling the definition and some of the properties of the $\BP$ operator, let us define an analogous quantum version. First, we introduce the notion of quantum reductions.

\begin{definition}[Quantum reductions] Let $A = (A_\textup{\sc yes},A_\textup{\sc no})$ and $B = (B_\textup{\sc yes},B_\textup{\sc no})$ be promise problems and $q(n) : \mathbb{N} \rightarrow \mathbb{N}$ be some polynomial. Let $\mathcal{A}$ be a quantum algorithm that takes as input the bit string representation of the instance of $A$, denoted as $x$ with $|x|=n$, applies a uniformly generated polynomial-time quantum circuit $V$ to $\ket{x}$ and a polynomial number of workspace qubits all initialised in $\ket{0}$, and measures $q(n)$ designated output qubits in the computational basis. Let $z \in \{0,1\}^{q(n)}$ be the outcome of this measurement. We say that $A \leq_q B$ if there exists a quantum algorithm $\mathcal{A}$ as described above such that
\begin{itemize}
    \item Completeness: if $x \in A_\textup{\sc yes} \Rightarrow \Pr_{z} [z \in B_\textup{\sc yes}] \geq 2/3$,
    \item Soundness: if $x \in A_\textup{\sc no} \Rightarrow \Pr_{y} [z \in B_\textup{\sc no}] \geq 2/3$,
\end{itemize}
where the probability is taken over the outcome strings $z$ of the quantum algorithm $\mathcal{A}$.
\end{definition}
Similar to randomised reductions, quantum reductions are also not transitive. Having formally defined our notion of quantum reduction, we can define the $\BQ$-operator in a similar way as the $\BP$-operator as per~\cref{def:BP} and~\ref{def:BP_alt}.

\begin{definition}[$\BQ$-operator] 
Let $\mathcal{C}$ be any class of (promise) problems. The class $\BQ \cdot \mathcal{C}$ consists of all (promise) problems $A = (A_\textup{\sc yes}, A_\textup{\sc no})$ that can be reduced to a complete (promise) problem $B = (B_\textup{\sc yes}, B_\textup{\sc no}) \in \mathcal{C}$ by a polynomial-time quantum reduction with success probability $\geq 2/3$, i.e.
\begin{align*}
    \BQ \cdot \mathcal{C} = \{P : P \leq_q B\}.
\end{align*}    
\label{def:BQ}
\end{definition}

Similarly, we can also provide an alternative definition, which can be viewed as feeding randomness generated by a quantum algorithm to be used by the verifier in the corresponding class. Crucially, the quantum algorithm's output is allowed to depend on the input $x$.

\begin{subdefinition}
    Let $\mathcal{C}$ be any class of (promise) problems. Then $\BQ \cdot \mathcal{C}$ consists of all (promise) problems $A = (A_\textup{\sc yes},A_\textup{\sc no})$ for which there exists a (promise) problem $B = (B_\textup{\sc yes}, B_\textup{\sc no}) \in \mathcal{C}$, and a polynomial-time quantum algorithm $\mathcal{A}$ such that for all $x$, $|x| = n$, 
\begin{itemize}
    \item Completeness: if $x \in A_\textup{\sc yes}$, then $\Pr_{y} [(x,z) \in B_\textup{\sc yes}] \geq 2/3$,
    \item Soundness: if $x \in A_\textup{\sc no}$, then $\Pr_{y} [(x,z) \in B_\textup{\sc no}] \geq 2/3$,
\end{itemize}
where $z \in \{0,1\}^{p(n)}$ is the outcome of a measurement in the computational basis of some designated $p(n)$-qubit register used by the quantum algorithm $\mathcal{A}(x)$.
\label{def:BQ_alt}
\end{subdefinition}

From the above definitions, it is clear that $\BQ \cdot \Pclass = \BQP$, just as $\BP \cdot \Pclass = \BPP$, as expected. By the same argument as before, when $\mathcal{C} \supseteq \Pclass$, both definitions are equivalent.

Let us now show that the same properties of the $\BP$-operator (\cref{prop:prop_BP}) also hold for the $\BQ$-operator. We will need the following amplification lemma, which can be proven in the same way as it was for the $\BP$-operator.

\begin{lemma}[Probability amplification for the $\BQ$-operator] 
If  $\mathcal{C}$ is a class closed under majority reducibility, then for every (promise) problem $A \in \BQ \cdot \mathcal{C}$ and every polynomial $p : \mathbb{N} \rightarrow \mathbb{N}$, there exists a (promise) problem $C \in \mathcal{C}$ such that for all $x$, $|x| = n$:
\begin{itemize}
    \item If $x \in A_\textup{\sc yes}$, then $\Pr_{z} [ (x,z) \in C_\textup{\sc yes}] \geq 1 - 2^{-p(n)}$,
    \item If $x \in A_\textup{\sc no}$, then $\Pr_{z} [ (x,z) \in C_\textup{\sc no}] \geq 1 - 2^{-p(n)}$,
\end{itemize}
where $z$ is the output of a quantum algorithm $\mathcal{B}$. 
\label{lem:prop_ampl}
\end{lemma}

\begin{proof}
Let $A = (A_\textup{\sc yes}, A_\textup{\sc no},A_\textup{\sc inv}) \in \BQ \cdot \mathcal{C}$ and $\mathcal{A}$ be the quantum algorithm used in the reduction. Then, there exists a problem $B = (B_\textup{\sc yes}, B_\textup{\sc no}, B_\textup{\sc inv}) \in \mathcal{C}$ such that
\begin{align*}
    &x \in A_\textup{\sc yes} \Rightarrow \Pr_{y}[(x,y) \in B_\textup{\sc yes}] \geq 2/3,\\
    &x \in A_\textup{\sc no} \Rightarrow \Pr_{y}[(x,y) \in B_\textup{\sc no}] \geq 2/3,
\end{align*}
where $y$ is a measurement outcome of $\mathcal{A}(x)$. Define $z = (z_1,\dots,z_k)$,  $|z| = cp(n)^2$, for some sufficiently large constant $c>0$, where each $z_i$ is a measurement outcome of a independent run of $\mathcal{A}$. By a standard Chernoff-bound argument, we have
\begin{align*}
   &x \in A_\textup{\sc yes} \Rightarrow \Pr_z[\text{the majority of } z_i \text{ satisfies } (x,z_i)\in B_\textup{\sc yes}] \geq 1 - 2^{-p(n)},\\
   &x \in A_\textup{\sc no} \Rightarrow \Pr_z[\text{the majority of } z_i \text{ satisfies }  (x,z_i)\in B_\textup{\sc no}] \geq 1 - 2^{-p(n)}.
\end{align*}
Now let $C = (C_\textup{\sc yes}, C_\textup{\sc no}, C_\textup{\sc inv})$ be a promise problem defined as:
\begin{itemize}
    \item $(x,z) \in C_\textup{\sc yes}$ if the majority of $z_i$ satisfies $(x,z_i) \in B_\textup{\sc yes}$,
    \item $(x,z) \in C_\textup{\sc no}$ if the majority of $z_i$ satisfies $(x,z_i) \in B_\textup{\sc no}$,
    \item $(x,z) \in C_\textup{\sc inv}$ otherwise.
\end{itemize}

Since $\mathcal{C}$ is closed under majority reducibility, we have that $C \in \mathcal{C}$ because $C \leq_\text{maj}^p B$, by constructing $(x,z)$ from the sequence $(x,z_1),\dots,(x,z_k)$ and using $B \in \mathcal{C}$. Therefore,
\begin{itemize}
    \item If $x \in A_\textup{\sc yes}$, then $\Pr_{z} [ (x,z) \in C_\textup{\sc yes}] \geq 1 - 2^{-p(n)}$,
    \item If $x \in A_\textup{\sc no}$, then $\Pr_{z} [ (x,z) \in C_\textup{\sc no}] \geq 1 - 2^{-p(n)}$.
\end{itemize}

The quantum algorithm $\mathcal{B}$ runs $k$ parallel executions of $\mathcal{A}$ and measures all $k$ designated output qubits to generate the instance $(x,z)$.
\end{proof}

\begin{restatable}{proposition}{propBQ} 
For all classes $\mathcal{C}, \mathcal{D}$, the following holds:
    \begin{enumerate}[label=(\roman*)]
        \item $\mathcal{C} \subseteq \mathcal{D}$ implies $\BQ \cdot \mathcal{C} \subseteq \BQ \cdot \mathcal{D}$.
        \item $\BP \cdot \mathcal{C} \subseteq \BQ \cdot \mathcal{C}$.
        \item $\BQ \cdot \BQ \cdot \mathcal{C} = \BQ \cdot \mathcal{C}$ if the class $\mathcal{C}$ is closed under majority reducibility.
        \item $\BQ \cdot \mathcal{C} \subseteq \BQP^{\mathcal{C}}$.
    \end{enumerate}
    \label{prop:prop_BQ}
\end{restatable}

\begin{proof} 
Points (i), (ii), and (iv) follow directly from the definition of the $\BQ$-operator.  For (iii), we use the fact that if $\mathcal{C}$ is closed under majority reducibility, then $\BQ \cdot \mathcal{C}$ is also closed under majority reducibility by~\cref{lem:prop_ampl}. We will only give the argument for the case when $x \in A_\textup{\sc yes}$, since a similar argument holds for $x \in A_\textup{\sc no}$.  Let $A = (A_\textup{\sc yes},A_\textup{\sc no})$ be any problem in $\BQ \cdot \BQ \cdot \mathcal{C}$. Hence, there exists a (promise) problem $B \in \BQ \cdot \mathcal{C}$, and a polynomial-time quantum algorithm $\mathcal{A}$ such that for all $x$, with $|x| = n$, we have that
\begin{align*}
x \in A_\textup{\sc yes} \Rightarrow \Pr_{y} [(x,y) \in B_\textup{\sc yes}] \geq c(n),
\end{align*}
where $y \in \{0,1\}^{p(n)}$ is the output of the measurement on $\mathcal{A}(x)$, and $c(n) = 1 - 2^{-p(n)}$ for some polynomial $p(n)$. Since $B \in \BQ \cdot \mathcal{C}$, there must also be a problem $C \in \mathcal{C}$ for which
\begin{align*}
    (x,y) \in B_\textup{\sc yes} \Rightarrow \Pr_{z} [(x,y,z) \in C_\textup{\sc yes}] \geq c(n),
\end{align*}
where $z$ is the output of some other polynomial-time quantum algorithm $\mathcal{A}'$ that takes $(x,y)$ as input.  Let $\mathcal{A}''$ be the quantum algorithm that runs $\mathcal{A}$ on input $x$, measures to obtain $y$, and then runs $\mathcal{A}'$ on $(x,y)$ to measure $z$ and obtain the instance $(x,y,z)$.  We then have
\begin{align*}
x \in A_\textup{\sc yes} \Rightarrow \Pr_{y,z} [(x,y,z) \in C_\textup{\sc yes}] =  \Pr_{z} [(x,y,z) \in C_\textup{\sc yes} | (x,y) \in B_\textup{\sc yes}] \cdot \Pr_{y} [(x,y) \in B_\textup{\sc yes}] \geq c^2(n),
\end{align*}
which can be easily made $\geq 2/3$ for a suitable choice of $p(n)$.  Since $C \in \mathcal{C}$, we have that $A \in \BQ \cdot \mathcal{C}$, showing that $\BQ \cdot \BQ \cdot \mathcal{C} \subseteq \BQ \cdot \mathcal{C}$ when $\mathcal{C}$ is closed under majority reducibility. The other direction of the inclusion holds trivially since $\mathcal{C} \subseteq \BQ \cdot \mathcal{C}$ for any class $\mathcal{C}$.
\end{proof}

Points (i) and (iii) of~\cref{prop:prop_BQ} together imply that $\BQ\cdot \BQP = \BQP$, since $\BQP = \BQ \cdot \Pclass = \BQ \cdot \BQ \cdot P = \BQ \cdot \BQP $. Furthermore, this also means that $\BQ \cdot \Pclass = \BQ \cdot \BPP$ since   $\BQ \cdot \BPP = \BQ \cdot \BP \cdot \Pclass \subseteq \BQ \cdot \BQ \cdot \Pclass = \BQ \cdot \BQP = \BQP = \BQ \cdot \Pclass$ and $\BQ \Pclass \subseteq \BQ \cdot \BPP$ since $\Pclass \subseteq \BPP$, which should be expected as quantum algorithms can generate randomness by themselves. As before, taking $\mathcal{C} = \NP$ with $B = \mathsf{3}\text{-}\mathsf{SAT}$, $\BQ \cdot \NP$ now defines the class of all problems that have \textit{quantum} reduction to $\mathsf{3}\text{-}\mathsf{SAT}$. We have that $\PCP[\mO(\log(n)), \mO(1)] = \NP \subseteq \MA \subseteq \AM = \BP \cdot \NP \subseteq \BQ \cdot \NP$ and $\BQP  \subseteq \BQ \cdot \NP$, but that it is unclear how $\BQ \cdot \NP$ relates to $\QCMA$ or $\QMA$ (for both inclusion directions). 

It is also important to note a difference in classes obtained from the $\BP$- and $\BQ$-operator when it comes to oracles. For example, for the $\BP$-operator we have that $\BPP^{\bmO} = \BP \cdot \Pclass^{\bmO} $ for all oracle sets $\bmO$~\cite{regan1992closure}. However, we have that there exist oracles $\bmO$ such that $\BQP^{\bmO}  \neq \BQ \cdot \Pclass^{\bmO} $, for example when considering the oracle that encodes Simon's problem. Hence, when we discuss oracle separations we always consider $ \left(\BQ \cdot \mathcal{C} \right)^{\bmO} $, so that the quantum reduction has (unitary) access to the oracle as well. 

\section{Pulling the quantumness out of quantum-classical PCPs}
\label{sec:QCPCP}
In this section, we will study quantum-classical PCPs and connect them to the previously introduced $\BQ$-operator. Let us start by recalling the definition of quantum-classical probabilistically checkable proof systems as given in~\cite{weggemans2023guidable}.

\begin{definition}[Quantum-Classical Probabilistically Checkable Proofs ($\QCPCP$)~\cite{weggemans2023guidable}] 
Let $n\in \mathbb{N}$ be the input size and $p, q : \mathbb{N} \rightarrow \mathbb{N}$, $c,s : \mathbb{R}_{\geq 0} \rightarrow \mathbb{R}_{\geq 0} $ with $c-s >0$. A promise problem $A = (A_\textup{\sc yes},A_\textup{\sc no})$ has a $(p(n),q(n),c,s)$-$\QCPCP$-verifier if there exists a $\Pclass$-uniform family of quantum verifier circuits $U = \{V_n : n \in \mathbb{N}\}$, each of which acts on an input $\ket{x}$ and a polynomial number of ancilla qubits, plus an additional bit string $y \in \{0,1\}^{p(n)}$ from which it is allowed to read at most $q(n)$ bits (non-adaptively), followed by a measurement of the first qubit, after which it accepts if and only if the outcome is $\ket{1}$, and satisfies:
\begin{itemize}
    \item If $x\in A_\textup{\sc yes}$, then there is a proof $y$ such that the verifier accepts with probability at least $c$,
    \item If $x\in A_\textup{\sc no}$, then for all proofs $y$ the verifier accepts with probability at most $s$.
\end{itemize}
A promise problem $A = (A_\textup{\sc yes},A_\textup{\sc no})$ belongs to $\QCPCP_{c,s}[p,q]$ if it has a $(p(n) , q(n),c,s$-$\QCPCP$ verifier. If $p(n)=\poly(n)$, $c=2/3$, and $s=1/3$, we simply write $\QCPCP[q]$. 
\label{def:QCPCP}
\end{definition}

In many settings, one allows for a more powerful quantum access model to classical strings: instead of making queries to a single location of the proof, one can also define a slightly more general version of quantum-classical probabilistically checkable proofs by allowing the entries of the proof string to be read in \emph{superposition}:
\begin{subdefinition}[$\QCPCP_Q$ ($\QCPCP$ with quantum queries)]  A $\QCPCP_Q$ protocol is just as a $\QCPCP$ protocol, but with the quantum verifiers $V_n$ given quantum access to the proof $y$, i.e., it takes one quantum query to implement the unitary $U_y$ defined as
\begin{align*}
    U_y : \ket{i}\ket{a} \mapsto \ket{i} \ket{a \oplus y_i},
\end{align*}
where $y_i$ is the value of the $i$'th bit of $y$. A promise problem $A = (A_\textup{\sc yes},A_\textup{\sc no})$ belongs to $\QCPCP_{Q,c,s}[p,q]$ if it has a $(p(n) , q(n),c,s)$-$\QCPCP_Q$ verifier. If $p(n)=\poly(n)$, $c=2/3$, and $s=1/3$, we simply write $\QCPCP_Q[q]$. 
\label{def:QCPCP_Q2}
\end{subdefinition}

\subsection{A threshold problem for multi-linear polynomials}
A \emph{multi-linear polynomial} of degree $d$ in $N$ variables with complex coefficients $\beta_S \in \mathbb{C}$ is a function $P : \mathbb{C}^N \rightarrow \mathbb{C}$ of the form
\begin{align}
    P(y) = \sum_{S \subseteq [N], |S| \leq d} \beta_S \prod_{i \in S} y_i,
\end{align}
where $y = (y_1, y_2, \ldots, y_N)$ is a vector of $N$ variables with $ y_i \in \mathbb{C}$ for all $ i \in [N]$, $\beta_S$ are the coefficients of the polynomial, and the sum is taken over all subsets $ S $ (including the empty subset $\emptyset$) of the index set $[N] = \{1, 2, \ldots, N\}$ where $ |S| \leq d $.  We will be interested in the setting where $y_i \in \{0,1\}$ and the coefficients are real, i.e.~$\beta_S \in \mathbb{R}$.

We now introduce the following decision problem, which we call the \emph{Multi-linear Polynomial Threshold Problem}.

\begin{definition}[Multi-linear polynomial threshold problem] Let $P : \{0,1\}^N \rightarrow \mathbb{R}$ be a multi-linear polynomial of degree $ d $ in $N$ variables with real coefficients $\{\beta_S\}$. Suppose that for all $y \in \{ 0,1\}^N$ we have $P(y) \in D \subset [0,1]$ for some finite, evenly spaced set $D$ with $\log_2(|D|)\leq \poly(N)$. Suppose that for all $y \in \{ 0,1\}^N$ we have $P(y) \in D \subset [0,1]$ for some finite, evenly spaced set $D$ with $\log_2(|D|)\leq \poly(N)$. Given as an input polynomial-size classical descriptions of $\{\beta_S\}$, $D$ and some $a \in D$, decide whether there exists a $y \in \{0,1\}^N$ such that $ P(y) \geq a$, or for all $y \in \{0,1\}^N$ we have $ P(y) < a$.
\label{def:mlptp}
\end{definition}
\noindent Clearly, the above problem is a valid decision problem and is contained in $\NP$. Note that it is only a decision problem because the function takes on only values from a finite set of numbers (which happens because of the constraint that all function values come from a discrete set $D$), which ensures that the promise $<a$ can alternatively be written as $\leq b$ for some $b-\varepsilon$, with $\varepsilon > 0$ being the spacing between the different numbers in the set $D$.

We will be interested in multi-linear polynomials with a constant degree and with a bounded range, i.e.~$\norm{P}_\infty \leq 1$, where $\norm{P}_{\infty} = \sup_y \abs{P(y)}$ denotes the uniform norm. Since we will only have indirect access to the polynomial, there is no a priori restriction on what the coefficients look like, nor even whether they can be accurately approximated using an efficient bit representation. Nevertheless, it is not too hard to show that if the polynomial itself is known to be bounded, all coefficients are also bounded, as shown in the following lemma.

\begin{restatable}{lemma}{lemboundcoef}
Let $P : \{0,1\}^N \rightarrow \mathbb{R}$ be a multi-linear polynomial of degree $d$ in $N$ binary variables with real coefficients $\{\beta_S\}$. Then for any coefficient $\beta_S$, it holds that
\begin{align*}
    |\beta_S| \leq \left(1 + (1 + 2^d)^{d-1}\right) \norm{P}_\infty .
\end{align*}
\label{lem:lemboundcoef}
\end{restatable}

\noindent The proof can be found in~\cref{app:omittedproofs}. Hence, when $\norm{P}_\infty \leq 1$ and $d$ is constant, the coefficients are also bounded in absolute value by some constant. This implies that all coefficients can be specified up to inverse exponential precision using a polynomial number of bits, and that given such a specification, $P(y)$ can also be exactly represented using a finite amount of bits.

\begin{lemma}
Let $P(y)$ be a multi-linear polynomial of degree $d$ in $N$ binary variables with real coefficients $\{\beta_S\}$, where each $\beta_S$ is given in $k$ bits of precision. Then we have that $P(y)$ can be represented exactly using at most $\log_2 \left( \binom{N}{d} \right)  + k$ bits.
\label{lem:bits_precision}
\end{lemma}
\begin{proof}
Since each $\beta_S$ is represented using $k$ bits, it can take up at most $2^k$ different values. Since $P(y)$ is always a sum over different choices of coefficients as the individual monomials which make up the polynomial can only be $0$ or $1$, and there are at most $\binom{N}{d}$ of them, we have that the number of values it can take is upper bounded by $\binom{N}{d} 2^k$, which means that $\log_2 \left( \binom{N}{d} \right)  + k$ bits suffice to exactly describe all possible values for $P(y)$.
\end{proof}
Thus, whenever the coefficients $\beta_S$ are specified using only a polynomial number of bits, the set $D$ is at most exponentially big.

\subsection{The polynomial method}
We will briefly review the ideas behind the polynomial method of~\cite{beals2001quantum} and show how it connects to quantum-classical PCPs with quantum queries. Let $\ket{0^m}$, where $m \geq \lceil \log N \rceil$, be a fixed initial state. The output state of a $q$-query quantum algorithm with query access to an input $y \in \{0,1\}^N$ in the form of a unitary $U_y : \ket{i} \ket{a} \rightarrow \ket{i} \ket{a \oplus y_i}$, with $a \in \{0,1\}$, can be written as (implicitly tensoring the $O_y$ query operations with identities)
\begin{align}
    \ket{\psi_q (y) } = U_q O_y U_{q-1} O_y \dots O_y U_1 O_y U_{0} \ket{0^m}.
    \label{eq:psi_q}
\end{align}
Since this state only depends on the input $y$ through the $q$ query operations, it is shown in~\cite{beals2001quantum} that $\ket{\psi_q (y) }$ can be written as 
\begin{align*}
    \ket{\psi_q (y) } = \sum_{z \in \{0,1\}^m} \alpha_z(y) \ket{z},
\end{align*}
where each $\alpha_z(y)$ is a multi-linear complex-valued polynomial in $y$ of degree at most $q$. For the output probability of measuring a single designated output qubit in $\ket{1}$, denoted as $P(y)$, we then have
\begin{align}
    P(y) &= \norm{\left(\ketbra{1} \otimes \mathbb{I}\right) \ket{\psi_q (y)}}^2_2 \label{eq:P_and_psiq} \\ 
         &=\sum_{z \in \{1\} \times \{0,1\}^{m-1}} |\alpha_z(y)|^2 \nonumber \\
         &= \sum_{S \subseteq [N],|S|\leq 2q} \beta_S \prod_{i \in S} y_i, \label{eq:multilin_polynom}
\end{align}
where the right-hand side is a multi-linear polynomial of degree $2q$ with \textit{real} coefficients $\beta_S$. 

The original application of the polynomial method is to prove lower bounds for quantum query algorithms, as the existence of a $T$-query quantum algorithm to compute a function $f$ implies the existence of a degree-$2T$ polynomial that approximates $f$. In the context of a $\QCPCP_Q[q]$-verifier, the only difference is that the proof $y \in \{0,1\}^N$ is now the string being queried, and the input $x \in \{0,1\}^n$ is used as part of the initial state for the quantum algorithm, i.e., $\ket{x} \ket{0^{m-n}}$, where $m \geq n + \lceil \log_2 N \rceil$. Therefore, the polynomial method directly implies the following lemma, as for a fixed input $x$, we have that $x$ can be absorbed into the unitary $U_0$ from~\cref{eq:psi_q}.

\begin{lemma} 
Let $A = (A_\textup{\sc yes},A_\textup{\sc no})$ be a promise problem in $\QCPCP_Q [p(n),q(n)]$. For a fixed input $x$, with $|x|=n$, let $V_x$ be the corresponding $\QCPCP_Q[p(n),q(n)]$ verifier with quantum query access to a proof $y \in \{0,1\}^{p(n)}$ and $x$ hardcoded into the circuit. Then, the probability that $V_x$ accepts is a multi-linear polynomial in $y$ of degree at most $2q(n)$.
\label{lem:QCPCP_polynomial}
\end{lemma}

A crucial point is that for the quantum-classical PCPs we consider, the string length $N$ is only \textit{polynomial} in the input size $n$ (whereas for Boolean functions typically considered in query complexity, we have $N = 2^n$). Therefore, the polynomial that describes the acceptance probability of the $\QCPCP_Q [q]$ verification circuit, given an input $x$, has an efficient classical description (provided the coefficients are specified up to a certain number of bits) whenever $q$ is constant. We will later see that this allows us to \textit{learn} an approximation of the polynomial which characterizes the proof input/output behaviour of the $\QCPCP_Q[q]$ verifier, given only access to a description of the input $x$ and the verification circuit. This learning algorithm never has to query the actual $U_y$; instead, it runs the $\QCPCP_Q[q]$ verifier for a large number of predetermined settings of ``fake'' proofs $y^{S}$, which are unrelated to the actual proof $y$. For any $y^{S}$, with $y^{S} \in \{0,1\}^{\poly(n)}$, the unitary $U_{y^S}$ can be efficiently implemented given the full description of $y^S$.\footnote{For a proof index $i$, if $y^{S}_i = 1$, implement an $X$ gate controlled by $\ket{i}$; otherwise, do nothing. This can be done for all $i \in [N]$ with $N = |y| = \poly(n)$, meaning that the overall circuit implementation will be efficient.}

\subsection{The quantum reduction}
We will now show that for a constant number of queries, it is possible to apply a quantum reduction to transform a $\QCPCP_Q[q]$-verification problem into a multi-linear polynomial threshold problem as per~\cref{def:mlptp}. The algorithm used in this reduction is given in~\cref{alg:reduction}. The key idea is to employ a learning algorithm that learns all coefficients of the polynomial (up to degree $2q$) ``from the ground up.'' Specifically, we will show that each coefficient $\beta_S$ for $|S|\leq l$ can be expressed in terms of a simple estimation procedure, relying on previously estimated coefficients $\beta_{S'}$ where $S' \subset S$. The main technical work lies then in determining parameters for the reduction, ensuring that it runs in polynomial time, and performing the reduction in such a way that we reduce to a \emph{decision} instead of a \emph{promise} problem with high probability.

\begin{algorithm} \caption{Quantum reduction from a $\QCPCP_Q$-verifier to a multi-linear polynomial} \label{alg:reduction}

\textbf{Input:} A classical description of a $\QCPCP_{Q,c,s}[p(n),q]$-verification circuit $V_x$ with input $x$ hardcoded into it, completeness and soundness parameters $c,s$.\\

\textbf{Output:} Classical descriptions of all coefficients $\hat{\beta}_S$ from a polynomial $\hat{P}(y)$, a threshold parameter $a$ and a set $D$.\\

\textbf{Algorithm:}
\begin{enumerate}
    \item Set
\begin{align*}
    &\epsilon := \frac{1}{16 (c-s) p(n)^{2q}(1+ 2^{2q-1} \left( p(n)^{2q}\right)^{2q} },\\
    &T:= \left\lceil  \left(\left(2^{\lceil \log_2 (1/(2\epsilon)+1) \rceil} -1 \right) \right)^2 \log \left(\frac{(2q+1) p(n)^{2q}}{\delta }\right) \right\rceil\\
    &l :=\lceil \log_2(1/(8\epsilon) +1) \rceil.
\end{align*}

    \item For $S \in \{S | S \subseteq [p(n)], |S| \leq 2q\}$:
    \begin{enumerate}
        \item[] For a step $S$, assuming that we have already stored $\hat{\beta}_S'$ for all $S' \subset S$.
        \item Set $y$ such that $y_i = 1$ if and only if $i \in S$.
        \item For $t \in [T]$, prepare $\ket{\psi_q (y)}$ and measure the first qubit in the computational basis. For random variable $X_t$, set $X_t = 1$ when the outcome was $\ket{1}$ and $X_t = 0$ when the outcome was $\ket{0}$. Let
        $
            \hat{X} = \frac{1}{T} \sum_{t \in [T]} \hat{X}_t
       $ and keep only the first $l$ description bits of $\hat{X}$. 
        \item Compute
        \begin{align*}
            \hat{\beta}_{S} &= \hat{X} - \sum_{S' \subset S} \hat{\beta}_{S'},
        \end{align*}
        and store $\hat{\beta}_{S} $.
    \end{enumerate}
    \item Output  
\begin{align*}
    &\hat{P}(y) = \sum_{S \subseteq [p(n)],|S|\leq 2q} \hat{\beta}_S \prod_{i \in S} y_i, \\
    &D:= G_{ 2^{2ll+1}},  \\
    &a:= \text{argmin}_{d\in D} \abs{d-(c+s)/2}.
\end{align*}
\end{enumerate}
\end{algorithm}

\begin{theorem} Let $q \in \mathbb{N}$ be some constant, and $c-s \geq 1/\poly(n)$. Let $A = (A_\textup{\sc yes},A_\textup{\sc no})$ be a promise problem in $\QCPCP_{Q,c,s}[p(n),q)]$, $x$, $|x|=n$, be an instance and $V_x$ be the circuit used by the $\QCPCP_Q[q)]$-verifier with the input $x$ hardcoded into the circuit, making $q$ quantum queries to a proof $y \in \{0,1\}^{p(n)}$. Then for any $\delta = \Omega(1/\exp(n))$, there exists a quantum polynomial-time reduction to a multi-linear polynomial threshold problem with degree $d = 2q$ and $\log_2 (|D|) = \poly(n)$, which succeeds with probability $1-\delta$.
\label{thm:reduction}
\end{theorem}

\begin{proof}

\phantom{.}
\paragraph{The minimum precision.} First, we need to establish what $a$ and the spacing between the elements of $D$ needs to be at the minimum. Denote $\delta_D$ for an upper bound on the allowed spacing between elements of $D$. Then we need to pick $a$ and $\delta_D$ such that $ a \in D$ and $a + \delta_D \leq c$ and $a - \delta_D \geq s$, which means that if the reduction succeeded we have successfully reduced to a decision problem. Hence, it suffices to have $\delta_D \leq (c-s)/4 = 1/\poly(n)$, and to set $a$ sufficiently close to $\frac{c+s}{2}$ (we will set $a$ and $D$ at the end of the proof).

\paragraph{The reduction.} We will proceed by proving that the reduction given by~\cref{alg:reduction} works in principle. By~\cref{lem:QCPCP_polynomial}, we have that the probability the $\QCPCP_Q[q]$-verifier $V_x$, with an input $x$ hardcoded into it, accepts a proof $y$ is given by
\begin{align*}
   P(y) =  \sum_{S \subseteq [p(n)],|S|\leq 2q} \beta_S \prod_{i \in S} y_i.
\end{align*}
With a given $S$ we associate a string $y^{S} \in \{0,1\}^{p(n)}$ such that $y^{S}_i = 1$ if and only if $i \in S$. For this string $y^{S}$ we can efficiently construct a ``fake'' proof unitary $U_{y^{S}}$, and consider the action of the $\QCPCP_Q$-verififier $V_x$ when given query access to $U_{y^{S}}$. We can express its acceptance probability as
\begin{align}
    P(y^{S}) &= \sum_{S' \subseteq [p(n)],|S'|\leq 2q} \beta_{S'} \prod_{i \in S} y^{S'}_i \nonumber\\
    &= \sum_{S' \subset S} \beta_{S'} + \beta_{S},\label{eq:sum_of_S}
\end{align}
keeping only the non-zero terms. Write $\Pi_1 = \ketbra{1} \otimes \mathbb{I}$ for the projection on the output qubit being in $\ket{1}$. Since we have that $P(y^S) = \norm{\Pi_1 \ket{\psi_q (y^S)}}^2_2$ by~\cref{eq:P_and_psiq}, we have that~\cref{eq:sum_of_S} can be rewritten as
\begin{align*}
    \beta_{{S}} &= \norm{\Pi_1 \ket{\psi_q (y^S)}}^2_2 - \sum_{S' \subset S} \beta_{S'}.
\end{align*}
Each coefficient $\beta_S$ can be expressed in terms of other coefficients $\beta_{S'}$, where $S' \subset S$, and the probability that $V_x$ accepts the input $y^S$. Hence, if no errors occurred in the estimation of each of the $\beta_S$,~\cref{alg:reduction} would give a perfect description of the polynomial $P$. However, errors will be introduced as $\norm{\Pi_1 \ket{\psi_q (y)}}^2_2$ can only be estimated. To ease presentation, we will introduce several error parameters which will eventually combined into the single error parameter $\epsilon$ as specified in~\cref{alg:reduction}.

Assume that we have that (i) $\abs{\hat{\beta}_S - \beta_S} \leq \epsilon_1$ for all $S \subseteq [p(n)], |S|\leq 2q$, (ii) that the overall estimation succeeded with probability $\geq 1-\delta$ and (iii) we have already obtained an $a$ that satisfies the above criteria. Then we have that our estimated polynomial $\hat{P}$ satisfies
\begin{align}
    \norm{\hat{P}(y) - P(y)}_{\infty} 
     &\leq \sum_{S \subseteq [p(n)],|S|\leq 2q} \abs{\tilde{\beta}_S - \beta_S} \\
     &\leq p(n)^{2q} \epsilon_1, \nonumber
\end{align}
which meets the minimum required precision $\leq \delta_D$ for $\epsilon_1 := \frac{1}{4}(c-s)  p(n)^{-2q}$. For our choice of $\epsilon_1$, we have that
\begin{itemize}
    \item $x \in A_\textup{\sc yes} \Rightarrow \Pr [ \exists y :  \hat{P}(y) \geq a] \geq  1-\delta$,
    \item $x \in A_\textup{\sc no} \Rightarrow  \Pr [\forall y :  \hat{P}(y) < a ] \geq  1-\delta$,
\end{itemize}
where the probability is taken over the outcome of the performed reduction. Let us now show that for our parameter choices, all the above three criteria are met.

\paragraph{Estimating $\norm{\Pi_1 \ket{\psi_q(y^S)} }_2^2$.} In the reduction we have to estimate $\norm{\Pi_1 \ket{\psi_q(y^S)} }_2^2 \in [0,1]$ for several settings of $y^S$, which is done in~\cref{alg:reduction} by preparing the state $\ket{\psi_q(y^S)}$ and performing a measurement in the computational basis of the first qubit. Suppose for now that we want to do the estimation of $\ket{\psi_q(y^S)}$ up to precision $\epsilon$, using the standard procedure for mean estimation. Let $X_i$ be the random variable for which $X_i =1$ if the outcome was $\ket{1}$ and $X_i = 0$ when the outcome was $\ket{0}$. Now let $\hat{X}$ be the random variable corresponding to 
\begin{align*}
    \hat{X} = \frac{1}{T} \sum_{t \in [T]} X_i,
\end{align*}
i.e.~the average of $T$ outcomes of $X_t$, which are independent and identically distributed. Since $0\leq X_i \leq 0$,  Hoeffding's inequality tells us that
\begin{align*}
    \Pr[\hat{X}-\mathbb{E}[X] \leq \epsilon] \geq 1-\exp(-2T\epsilon^2).
\end{align*}
To achieve overall success probability $1-\delta$, we require that
\begin{align*}
   1-\exp(-2T\epsilon^2) \geq 1-\delta,
\end{align*}
which is satisfied when
\begin{align}
    T \geq \frac{\log \left(\frac{1}{\delta}\right) }{2 \epsilon^2}.
    \label{eq:bound_T}
\end{align}

\paragraph{Estimation precision.}  First, we will find the minimum required precision $\epsilon_2$ of estimating each $\norm{\Pi_1 \ket{\psi_q(y^{S'})} }_2^2$ in order to achieve $\abs{\tilde{\beta}_S - \beta_S} \leq \epsilon_1$ for all $S \subseteq [p(n)], |S|\leq 2q$. Write $X_S = \norm{\Pi_1 \ket{\psi_q(y^{S})} }_2^2$ and $\hat{X}_S$ for its estimate. For some $\beta_S$, we have
\begin{align*}
    \abs{\beta_S - \hat{\beta_S}} &= \abs{X_S - \sum_{S' \subset S} \beta_{S'} - \left(\hat{X}_S - \sum_{S' \subset S} \hat{\beta}_{S'}\right)}\\
    &\leq \abs{X_S -\hat{X}_S} + \abs{  \sum_{S' \subset S} \hat{\beta}_{S'} - \sum_{S' \subset S} \beta_{S'}}\\
    &\leq \abs{X_S -\hat{X}_S} + \sum_{S' \subset S} \abs{ \hat{\beta}_{S'} -  \beta_{S'}},
\end{align*}
Since the error guarantee only depends on the cardinalities of each $S' \subset S$, i.e.~$|S'|$, we replace $\beta_{S'}$ with $\beta_i$ such that with $i=|S'|$. Let $l=|S|$. Using that  $\abs{X_S -\hat{X}_S}  \leq \epsilon_2$ for any $S$, we have that
\begin{align*}
 \abs{\beta_l - \hat{\beta_l}} &\leq \epsilon_2 + \sum_{S' \subset S} \abs{ \hat{\beta}_{S'} -  \beta_{S'}}\\
   &\leq \epsilon_2 +   \sum_{i=0}^{l-1} \binom{p(n)}{i} \abs{ \hat{\beta}_{i} -  \beta_{i}}.
\end{align*}
Define $B(i) = \abs{\beta_i - \hat{\beta_i}}$ for $i = \{0,1,\dots,l\}$. Since all terms are positive, this expression can be upper bounded by solving a recursion
\begin{align}
     B(l) &\leq \epsilon_2 +   \sum_{i=0}^{l-1} \binom{p(n)}{i} B(i)
     \label{eq:B(l)}
\end{align}
We will need the following lemma, that we prove in~\cref{app:omittedproofs}:
\begin{restatable}{lemma}{recursionlemma} Let $f(l)$ be a recursion given by
\begin{align*}
    f(l) =  x +  \sum_{i=0}^{l-1} a_i f(i)
\end{align*}
where $\{a_i\}$ are real coefficients and initial value $f(0) = x$. Then we have that
\begin{align*}
     f(l) = (1+\sum_{S \subseteq [l-1]} \prod_{i \in S} a_i ) x
\end{align*}
\label{lem:recursion}
\end{restatable}
We have that~\cref{eq:B(l)} is the form of~\cref{lem:recursion} with $a_i = \binom{p(n)}{i}$.  We bound each coefficient as
\begin{align*}
    \binom{p(n)}{i} \leq  p(n)^{2q}
\end{align*}
using that $i \leq 2q$. By~\cref{lem:recursion} we find
\begin{align*}
     \abs{\beta_{2q} - \hat{\beta_{2q}}} &\leq \epsilon_2 (1+\sum_{S \subseteq [2q-1]} \prod_{i \in S} p(n)^{2q} ) \leq  \epsilon_2 (1+ 2^{2q-1} \left( p(n)^{2q}\right)^{2q} ) \leq \epsilon_1. 
\end{align*}
when $\epsilon_2 = \epsilon_1/(1+ 2^{2q-1} \left( p(n)^{2q}\right)^{2q} ) $.

\paragraph{Success probability.}
To ensure that the overall reduction succeeds with probability $1-\delta$, we need the estimation of $\norm{\Pi_1 \ket{\psi_q(y^S)} }_2^2$ all up to accuracy $\epsilon$ succeeds on all instances. Let $1-\delta'$ be the success probability of each of these estimations. We can upper bound the total number of subsets $S$ as
\begin{align*}
    \sum_{l=0}^{2q} \binom{p(n)}{l} \leq (2q+1) p(n)^{2q}
\end{align*}
using the previously used upper bound on the binomial coefficient. This means that it suffices to have
\begin{align*}
    (1-\delta')^{ (2q+1) p(n)^{2q}} \geq 1-\delta \Leftrightarrow \delta' \leq 1-(1-\delta)^{\left(  (2q+1) p(n)^{2q} \right)^{-1}},
\end{align*}
for which we have that 
\begin{align*}
    1-(1-\delta)^{\left( (2q+1) p(n)^{2q}\right)^{-1}} \geq \frac{\delta}{ (2q+1) p(n)^{2q}},
\end{align*}
where we used the inequality $(1+x)^r \leq 1+rx$ for $x\geq -1$, $r \in [0,1]$. Hence, setting
\begin{align*}
    \delta' := \frac{\delta}{ (2q+1) p(n)^{2q}}
\end{align*}
suffices. 

\paragraph{Bits of precision and $D$.}
We now want to change our estimate to correspond to bits of precision. Specifically, we will pick an $\epsilon \leq \epsilon_2$ such that we can take the first $l$ bits of our estimate $\hat{X}$ and, conditioning on the estimation procedure succeeding, be sure that they are correct and be sure that we have $\leq \epsilon_2$ additive error. Recall that for any $S$, the random variable $X_i \in \{0,1\}$. Hence, when the estimate $\hat{X}$ is constructed out of $T$ trials, we have that the total outcome space is given by $\{0,\frac{1}{T},\frac{2}{T},\dots,1-\frac{1}{T},1\}=:G_T$. Assume that $T$ satisfies $T + 1 =2^k$ for some $k \in \mathbb{N}$, such that $G_t$ can be represented exactly using $k$ bits. If we want our estimate to have that the first $l \leq k$ bits are correct, this is equivalent to having
\begin{align*}
    \text{First $l$ bits of $\hat{X}$ are correct} \iff \epsilon_2 \leq \frac{1}{2\left(2^l-1\right)}.
\end{align*}
Hence, we have to find an $\epsilon$ below $\epsilon_2$ such that it becomes exactly expressible as the right-hand side of the above inequality. We can do so by setting
\begin{align*}
    \frac{1}{2\left(2^{\lceil \log_2 (1/(2\epsilon_2)+1) \rceil} -1 \right)} \geq \frac{\epsilon_2}{2(1+\epsilon_2)} \geq \epsilon_2/4 =:\epsilon,
    \end{align*}
so that at least $l=\lceil \log_2(1/(2\epsilon_2) +1) \rceil$ bits are correct. By~\cref{lem:lemboundcoef}, we have that each $\beta_S$ is contained in the interval
\begin{align*}
    \beta_S \in [-(1+(1+2^{2q})^{2q-1}), (1+(1+2^{2q})^{2q-1})],
\end{align*}
so we need at most $k = \log_2 \left( 2( 1+(1+2^{2q})^{2q-1}) \right) l +1$ bits to describe $\beta_S$ (since $l$ bits are used to described the interval $[0,1]$). Since $\hat{P}(y) \in [-\delta_D,1+\delta_D] $, we can set $D$ as
\begin{align*}
  D:=  G_{ 2^{2ll+1}} = G_{ 4 \lceil \log_2(1/(2\epsilon_2) +1) \rceil } 
\end{align*}
We only need to ensure that our value of $a$ is in the set. Since $\epsilon_2 < \delta_D$, setting $a :=  \text{argmin}_{d\in D} \abs{d-(c+s)/2} $ suffices to have $a \in (s,c)$ and $a \in D$. Clearly, we have $\log_2(|D|) = \poly(n)$, as desired.

\paragraph{Complexity and combining relevant parameters.} Finally, we combine all the above parameters and show that this results in a polynomial run time when $q$ is constant. For our final choice of $\epsilon$, we can set $T$ satisfying~\cref{eq:bound_T} for $\epsilon$ instead of $\epsilon_2$, and we find
\begin{align*}
&\epsilon := \frac{1}{16 (c-s) p(n)^{2q}(1+ 2^{2q-1} \left( p(n)^{2q}\right)^{2q} }  =1/\poly(n)\\
    &T:= \left\lceil  \left(\left(2^{\lceil \log_2 (1/(2\epsilon)+1) \rceil} -1 \right) \right)^2 \log \left(\frac{(2q+1) p(n)^{2q}}{\delta }\right) \right\rceil = \poly(n),
\end{align*}
which gives an overall runtime of the order of $\mO(p(n)^{2q} T) = \poly(n)$.
\end{proof}

\subsection{Implications}
Similar to what was observed in~\cite{buhrman2024quantum}, quantum reductions are a versatile tool to prove many properties of quantum PCPs. We will now prove several statements regarding quantum-classical PCPs, all of which to some extent utilize the quantum reduction from~\cref{thm:reduction}. 

\subsubsection{A new upper bound on quantum-classical PCPs}
\Cref{thm:reduction} directly implies the following corollary, as we give a randomized reduction to a problem in $\NP$ (note that this is \emph{not} the promise version).

\begin{corollary} For any constant $q\in \mathbb{N}$, we have
\begin{align*}
    \QCPCP_Q[q] \subseteq \BQ \cdot \NP.
\end{align*}
\label{cor:containment_BQNP}
\end{corollary} 
We remark that $\QCPCP[\mO(1)] \subseteq \BQ \cdot \NP$ also follows from~\cite[Theorem 25]{weggemans2023guidable}, where a similar proof was used to derive the weaker upper bound of $\BQP^{\NP[1]}$ (which is a $\BQP$ verifier that is allowed to make a single query to an $\NP$ oracle). A key difference, however, is that the proof in~\cite{weggemans2023guidable} does not remove the promise, which was acceptable for their purposes since all complexity classes (even the classical ones) were defined as promise classes. Finally, note that the inclusion $\BQ \cdot \NP \subseteq \QCPCP[\mO(1)]$ does not necessarily hold, as the quantum reduction might produce different $\NP$ problems, each requiring different witnesses. Thus, the prover may not know which proof to provide.

\subsubsection{Constant query hierarchy collapse, query equivalence and amplification}
We will use H\aa stad's PCP to show the constant query hierarchy collapse.

\begin{lemma}[\cite{HastadSome1997}] For every $\delta > 0$ and every decision problem $D \in \NP$ there is a PCP-verifier $V$ for $L$ making $3$ queries having completeness parameter $1 - \delta$ and soundness parameter at most $\frac{1}{2} + \delta$.
\label{lem:hastad_PCP}
\end{lemma}

\begin{proposition} For any $q\in \mathbb{N}$ constant and for any $c,s \geq 1/\poly(n)$, we have 
\begin{align*}
\QCPCP_{Q,c,s}[q] \subseteq \QCPCP_{1-\delta,1/2+\delta}[3],
\end{align*}
for any $\delta >0$ constant.
\label{prop:QCPCP_equiv}
\end{proposition}
\begin{proof}
From~\cref{thm:reduction} we know that for any promise problem $A = (A_\textup{\sc yes},A_\textup{\sc no})$ in $\QCPCP_{Q,c,s}[q]$ with input $x$ and completeness and soundness parameters $c,s$ with $c-s \geq 1/\poly(n)$, there exists a quantum reduction to a multi-linear polynomial threshold problem. Conditioned on the quantum reduction succeeding, this reduction is deterministic: since the multi-linear polynomial is learned up to a fixed number of bits of precision, it will always output the same polynomial. Since deciding if a multi-linear polynomial has a setting for which it evaluates to a value larger than $a$, or evaluates to a value smaller than $a$ is in $\NP$, it can be correctly decided by the PCP from~\cref{lem:hastad_PCP} with completeness $1-\delta_1$ and soundness $\frac{1}{2}+\delta_1$.  The $\QCPCP_{1-\delta,\frac{1}{2}+\delta}$ protocol naturally follows: the verifier performs the reduction with success probability $1-\delta_2$ to obtain the polynomial $\hat{P}$ and then uses the $\PCP$ protocol and the proof $y$ given by the prover to decide if there exists a $y$ such that $\hat{P}(y) \geq a$ or whether for all $y$ we have $ \hat{P}(y) <a$. We have that
\begin{itemize}
    \item $x \in A_\textup{\sc yes} \Rightarrow \Pr [ \QCPCP \text{ protocol accepts}] \geq  (1-\delta_1)(1-\delta_2) \geq 1-\delta$,
    \item $x \in A_\textup{\sc no} \Rightarrow  \Pr [ \QCPCP \text{ protocol accepts}] \leq  1-(1-\delta_1)(\frac{1}{2} -\delta_2) \leq \frac{1}{2} + \delta$,
\end{itemize}
when $\delta_1 = \delta_2 = \delta/2$. Since the reduction of~\cref{thm:reduction} runs in time polynomial in $\log(1/\delta_2)$, it will clearly run in polynomial time when $\delta_2$ is constant.
\end{proof}
There are two important remarks one can make with respect to~\cref{prop:QCPCP_equiv}.

\begin{remark}
   The idea behind~\cref{prop:QCPCP_equiv} also shows that w.l.o.g.~a quantum-classical PCP can be assumed to use a uniformly random distribution over (a subset of) the indices of the proof to be queried and then uses a single quantum circuit to determine the outcome, as there are constant-query non-adaptive PCPs that use a uniform random distribution (for example by using the $\NP$-hard gapped version of $3$-SAT, which can be achieved using Dinur's gap amplification procedure~\cite{Dinur2007thepcp}).
\label{rem:unif_PCP}
\end{remark}

\begin{remark}
 Even though~\cref{prop:QCPCP_equiv} uses a classical PCP construction, which is non-relativizing, the theorem itself \emph{does} relativize. The reason for this is that the quantum reduction does allow the $\QCPCP_Q[q]$-verifier to make queries to an additional oracle, so the learned polynomial will have the oracle calls implicitly encoded into it. The degree of the polynomial is independent of external oracle queries and only depends on the number of quantum queries to the proof.
 \label{rem:relativizing}
\end{remark}
Finally, we would like to point out that~\cref{prop:QCPCP_equiv} also has another immediate corollary: $\QCPCP$ with a constant number of adaptive queries has the same power as constant-query non-adaptive $\QCPCP$, as $\QCPCP_Q$ can be viewed as a generalization of adaptive $\QCPCP$. This rediscovers another result from~\cite{weggemans2023guidable}.

\subsubsection{An alternative characterization}
Using the observation in~\cref{rem:unif_PCP}, one can show a similar result to $\QCMA = \NP^{\BQP}$ (see~\cite{weggemans2023guidable}) for quantum-classical PCPs. Note that this is the \textit{two-sided} error version of the class $\PCP$, which contains the standard definition via error reduction.

\begin{proposition}
We have that $\QCPCP_Q[\mO(1)] = \PCP_{2/3,1/3}[\log n,\mO(1)]^{\BQP}$. 
\label{prop:alt_charac}
\end{proposition}

\begin{proof}
We prove both directions of the equality. 

\paragraph{$\PCP_{2/3,1/3}[\mO(\log n),\mO(1)]^\BQP \supseteq \QCPCP[\mO(1)]$.} 
Let $q \in \mO(1)$. By~\cref{rem:unif_PCP}, we can assume, without loss of generality (for sufficiently large $q$, depending on the completeness and soundness parameters), that the quantum-classical PCP selects $q$ indices uniformly at random and then runs a verifier $V$, determining acceptance or rejection based on a single output qubit. Thus, we can use the non-adaptive definition of $\QCPCP[q]$ as given in~\cref{def:QCPCP}. For any promise problem $A = (A_\textup{\sc yes},A_\textup{\sc no})$ in $\QCPCP[q]$, we apply weak error reduction, resulting in a $q' = \mO(q)$ number of queries, and obtain the following:
\begin{itemize}
    \item If $x \in A_\text{{\sc yes}}$, there exists a $y$ such that $\mathbb{E}_{i_1,\dots,i_{q'}} [\Pr [V \text{ accepts } (x,y) ] ] \geq 0.99$.
    \item If $x \in A_\text{{\sc no}}$, for all $y$, we have $\mathbb{E}_{i_1,\dots,i_{q'}} [\Pr [V \text{ accepts } (x,y) ] ] \leq 0.01$,
\end{itemize}
where the indices are selected uniformly at random from a subset of all possible indices (which is due to the weak error reduction). Denote this subset as $I'$. Hence, when picking a set of $q'$ indices from $I'$ uniformly at random, when $x \in A_\text{{\sc yes}}$, with probability $\geq 9/10$ these will yield an acceptance probability $\geq 9/10$. Similarly, for $x \in A_\text{{\sc no}}$, with probability $9/10$, these will have an acceptance probability $\leq 1/10$. The $\PCP_{2/3,1/3}[\mO(\log n),\mO(1)]^{\BQP}$-protocol $\mathcal{A}$ is as follows: pick uniformly random indices $i_1,\dots,i_{q'}$ from $I'$, and use the $\BQP$ oracle to simulate the $\QCPCP$ verifier, querying whether the acceptance probability for $(y_{i_1},\dots,y_{i_{q'}},x)$ is $\geq 9/10$ or $\leq 1/10$. We then have:
\begin{itemize}
    \item If $x \in A_\text{{\sc yes}}$, there exists a $y$ such that $\Pr[\mathcal{A} \text{ accepts }(x,y) ] \geq 9/10$.
    \item If $x \in A_\text{{\sc no}}$, for all $y$, we have $\Pr[\mathcal{A} \text{ accepts }(x,y)] \leq 1/10$.
\end{itemize}
which means that $A $ is in $\PCP_{2/3,1/3}[\mO(\log n),\mO(1)]^\BQP$.

\paragraph{$\PCP_{2/3,1/3}[\mO(\log n),\mO(1)]^\BQP \subseteq \QCPCP[\mO(1)]$.} 
Let $p_1,p_2:\mathbb{N} \rightarrow \mathbb{N}$ be polynomials, and let $q\in \mO(1)$ be a constant. Consider a deterministic polynomial-time verification circuit $M^{\Pi,y}(x,r)$ with inputs $x$ and a random bit string $r$ of length $\mO(\log n)$, allowed to query the proof string $y \in \{0,1\}^{p_1(n)}$ up to $q$ times, and a $\BQP$ oracle $\Pi$ up to $p_2(n)$ times as a black box. For every query $z^i \in (\Pi_\textup{\sc yes},\Pi_\textup{\sc no})$ made to $\Pi$, there exists a quantum algorithm $\mathcal{A}(z^i)$ such that
\begin{itemize}
    \item If $ \Pi(z^i) =1$, then $\Pr[\mathcal{A}(z^i) = 1] \geq 1-2^{-p_3(n)}$,
    \item If $\Pi(z^i) =0 $, then $ \Pr[\mathcal{A}(z^i) = 0] \geq 1-2^{-p_3(n)}$.
\end{itemize}
for any polynomially bounded function $p_3(n)$. Let $Z(r) = \{z^1,\dots,z^{p_2(n)}\}$ be the set of all query strings to $\Pi$ given a certain $r$, and let $I(r) = \{i_1,\dots,i_{q}\}$ likewise be the set of all indices where $y$ is supposed to be queried (again given $r$). For any problem $A = (A_\textup{\sc yes},A_\textup{\sc no})$ to be in $\PCP_{2/3,1/3}[\mO(\log n),\mO(1)]^\BQP $, we must have that
\begin{itemize}
    \item If $x \in A_\textup{\sc yes} $, then $ \exists y: \Pr[M^{\Pi,y}(x,r) =1 ] \geq 2/3$,
    \item If $x \in A_\textup{\sc no}$, then $\forall y: \Pr[M^{\Pi,y}(x,r) = 1] \leq 1/3$.
\end{itemize}
However, $M$ might also make invalid queries to the oracle $\Pi$. Let $Z_\textup{\sc inv}(r) \subseteq Z(r)$ be set of all invalid queries made given a $r$. With some abuse of notation, we write $M^{\Pi',y}(x,r,o)$ for the machine $M$ which now uses the oracle $\Pi'$ which, given a certain $r$, only makes the valid queries to $\Pi$ and replace the outcomes of invalid oracle calls to $\Pi$ by an additional input string $o$ of length $|Z_\textup{\sc inv}(r)|$. We must have that
\begin{itemize}
    \item If $x \in A_\textup{\sc yes} $, then $ \exists y, \forall o \in \{0,1\}^{|Z_\textup{\sc inv}(r)|} : \Pr[M^{\Pi',y}(x,r,o) =1 ] \geq 2/3$,
    \item If $x \in A_\textup{\sc no} $, then $\forall y, \forall o \in \{0,1\}^{|Z_\textup{\sc inv}(r)|} : \Pr[M^{\Pi',y}(x,r,o) = 1] \leq 1/3$.
\end{itemize}
This means that one can take an ``adversary'' view of the invalid queries, in the sense that whatever the oracle outputs should still satisfy the completeness and soundness condition to have that a problem is in the class. Now assume the limiting case where $|Z_\textup{\sc inv}(r)| = \emptyset$, i.e. $M^{\Pi,y}(x,r,o)$ only makes valid queries (since the simulation will never be wrong on invalid queries by the above argument). Then we have that all queries are simulated correctly by $\mathcal{A}(z_i)$ is given by
\begin{align*}
    (1-2^{-p_3(n)})^{p_2(n)} \geq 1-\frac{p_2(n)}{2^{p_3(n)}} \geq 0.9
\end{align*}
for $p_3(n) \geq \log(10 q(n ))$. Since this means that $\QCPCP$ can simulate $M^{\Pi,y}(x,r)$ with high accuracy, we have that
\begin{itemize}
    \item If $x \in A_\textup{\sc yes} $, then $ \exists y: \Pr[\mathcal{A}(x,y) \text{ outputs } 1 ] \geq 0.6 $,
    \item If $x \in A_\textup{\sc no} $, then $ \forall y: \Pr[\mathcal{A}(x,y) \text{ outputs } 1] \leq 0.3$,
\end{itemize}
which are separated $\Omega(1)$, putting the problem in $\QCPCP[\mO(1)] \subseteq \QCPCP_Q[\mO(1)] $ using standard error reduction.
\end{proof}

Since $ \QCPCP[\mO(1)] \subseteq \BQ \cdot \NP \subseteq \BQ \cdot \PCP_{2/3,1/3}[\mO(\log n),\mO(1)] \subseteq \BQP^{\PCP_{2/3,1/3}[\mO(\log n),\mO(1)]}$, the following corollary follows.
\begin{corollary}$
        \PCP_{2/3,1/3}[\mO(\log n),\mO(1)]^\BQP \subseteq \QCPCP[\mO(1)] \subseteq \BQP^{\PCP_{2/3,1/3}[\mO(\log n),\mO(1)]}$
        \label{cor:QCPCP_ULB}
\end{corollary}
Since this result also relativizes, this contrasts the results for $\NP$, where it is known that $\NP^\BQP  \not\subset \BQP^\NP$ relative to an oracle~\cite{aaronson2021acrobatics}. Since $\QCPCP[\mO(1)] \subseteq \QPCP[\mO(1)]$,~\cref{cor:QCPCP_ULB} also gives the first lower bound on quantum PCPs beyond the trivial lower bound of $\NP$ which follows from the classical PCP theorem.

\section{Oracle separations for quantum-classical PCPs}
In this section we will again focus on quantum-classical PCPs, but now in the presence of oracles. We will see that our result that $\QCPCP[\mO(1)] \subseteq \BQ \cdot \NP$ (which holds relative to all oracles, see~\cref{rem:relativizing}) implies the existence of a classical oracle relative to which the quantum-classical PCP conjecture is false. We will also study the setting where the number of queries is (poly-)logarithmic instead of constant. We show that relative to an oracle quantum-classical PCPs with quantum queries are more powerful than those with classical queries.

\subsection{The quantum-classical PCP conjecture is false relative to an oracle}
We will first argue that some established inclusions, used to provide our desired oracle separation, relativize. This way, we can directly employ the OR $\circ$ Forrelation oracle as described in~\cref{subsec:or_forr_oracle}.
\begin{lemma}
    For all $q \in \mathbb{N}$ and all oracles $\bmO$, we have that $\QCPCP[q]^{\bmO} \subseteq \QCPCP_Q[q]^{\bmO} \subseteq \left(\BQ \cdot \NP\right)^{\bmO} \subseteq \BQP^{\NP^{\bmO}}$.
\label{lem:contwithA}
\end{lemma}
\begin{proof}
   Let $\bmO$ be any oracle. The inclusion $\QCPCP[q]^{\bmO} \subseteq \QCPCP_Q[q]^{\bmO}$ trivially follows by a simulation argument. We have that $\QCPCP_Q[q]^{\bmO} \subseteq \BQP^{\NP^{\bmO}}$ follows from the fact that all oracle calls, which can w.l.o.g.~be viewed as quantum queries $U_f$ to $\bmO$, can be absorbed into the unitaries in $\ket{\psi_q (y)} = U_q O_y U_{q-1} O_y \dots O_y U_q O_y U_0 \ket{x}\ket{0^{m-n}}$, which shows that the proof of~\cref{thm:reduction} and thus $\QCPCP_Q[q]^{\bmO} \subseteq \left(\BQ \cdot \NP\right)^{\bmO}$ relativizes. For $\left(\BQ \cdot \NP\right)^{\bmO} \subseteq \BQP^{\NP^{\bmO}}$, This this is easily shown by noting that we can write $\left(\BQ \cdot \NP\right)^{\bmO}$ as the class of all problems $A = (A_\textup{\sc yes},A_\textup{\sc no})$ for which there exists a polynomial time deterministic Turing machine $M$ with oracle access to $\bmO$, and a polynomial-time quantum algorithm $\mathcal{A}^{\bmO}$ with (unitary) oracle access to $\bmO$ such that for all $x$, $|x| = n$, 
\begin{itemize}
    \item Completeness: if $x \in A_\textup{\sc yes} \Rightarrow \Pr_{z} [\exists y: M^{\bmO}(x,y,z) \text{ accepts}] \geq 2/3$,
    \item Soundness: if $x \in A_\textup{\sc no} \Rightarrow \Pr_{z} [\forall y: M^{\bmO}(x,y,z) \text{ rejects}] \geq 2/3$,
\end{itemize}
where $z \in \{0,1\}^{p(n)}$ is the measured output of the quantum algorithm $\mathcal{A}^{\bmO}(x)$. Note that, given some $x\in \{0,1\}^{n},z \in \{0,1\}^{p(n)}$ and a description of $M^{\bmO}(x,z)$, finding out whether $\exists y: M^{\bmO}(x,y,z)=1$ or $\forall y: M^{\bmO}(x,y,z)=0$ can be decided correctly by a $\NP^{\bmO}$ oracle. Since the class $\BQP^{\NP^{\bmO}}$ has oracle access to $\NP^{\bmO}$ and therefore also ${\bmO}$, it can run the quantum algorithm $\mathcal{A}^{\bmO}$ on input $x$ to obtain some $z$, and then use its $\NP^{\bmO}$ oracle to solve the corresponding decision problem. The overall procedure then succeeds with probability $\geq 2/3$, since it succeeds with probability $1$ if the reduction was performed successfully.
\end{proof}

These relativizing inclusions directly imply our desired oracle separation.

\orsepQCMA*
\begin{proof}
    This follows from~\cref{prop:OracleSep} and~\cref{lem:contwithA}, which shows that for all constant $q$, relative to all oracles ${\bmO}$ we have $
       \QCPCP[q]^{\bmO} \subseteq \QCPCP_Q[q]^{\bmO} \subseteq \BQP^{\NP^{\bmO}}$,   but that there exists an oracle ${\bmO}$ such that $
       \BQP^{\NP^{\bmO}} \not\supset \QCMA^{\bmO}$.   The fact that $\QCPCP_Q[q]^{\bmO} \subseteq \QCMA^{\bmO}$ holds trivially since the unitary $U_y$ can be implemented efficiently given the full description of a polynomially-sized proof $y$.
\end{proof}

\subsection{Oracle separations for logarithmic queries}
Whilst our results in~\cref{sec:QCPCP} indicate that in the constant query regime, quantum queries offer no advantage over classical queries, even in the presence of external oracles (see~\cref{rem:relativizing}), we will now prove that this changes when the number of queries is (poly-)logarithmic relative to an oracle.

We will rely on the following lemma, which shows that the Bernstein-Vazirani algorithm~\cite{bernstein1993quantum} can be used to decode $\mO(\log n)$ bits from a polynomially-sized classical string using only a single quantum query.

\begin{lemma}
Given any $\mO(\log n)$-sized bit string $x$, there exists a polynomially-sized classical proof $y$ such that a quantum algorithm can learn $x$ by making only a single quantum query to $y$.
\label{lem:BV_string_hiding}
\end{lemma}

\begin{proof}
This follows directly from the Bernstein-Vazirani algorithm~\cite{bernstein1993quantum}. The algorithm can learn a secret bit string $x$ of length $l = \mO(\log n)$, provided it has quantum oracle access to a function $f : \{0,1\}^l \rightarrow \{0,1\}$ defined by $f(z) = z \cdot x \mod 2$. For any such $x \in \{0,1\}^l$, the prover constructs a function $f$ satisfying this property, and sends a proof $y = (y_1,\dots,y_{2^l})$, where $y_{\bar{z}} = f(z)$ and $\bar{z}$ is the integer representation of the string $z$. Since $l \in \mO(\log n)$, we have $|y| = 2^{l} = \poly(n)$. By leveraging the fact that a phase query can be implemented at the same cost as a standard oracle query, the verifier can extract $x$ by following the Bernstein-Vazirani algorithm, making only a single quantum query to $y$.
\end{proof}

Next, we state the lower bound lemma that lower bounds the query complexity of the OR function, given access to some additional classical bits in assistance of the verification.
\begin{lemma} Suppose we are given oracle access to an $n$-qubit phase oracle ${\bmO_f}$, and want to decide which of the following holds:
\begin{enumerate}[label=(\roman*)]
    \item There exists an $n$-bit string $x^*$ such that ${\bmO_f} \ket{x^*}  = -\ket{x^*} $
    \item ${\bmO_f} \ket{x}  = \ket{x} $ for all $x$.
\end{enumerate}
Then even if we have $q$ classical bits in support of case (i), we still need to make
\begin{align*}
 \Omega \left( \sqrt{ \frac{2^n}{2^q}} \right).
\end{align*}
queries to decide between cases (i) and (ii) with bounded probability of error. 
\label{lem:lb_classical_proof}
\end{lemma}
\begin{proof}
Let $C_1 > 0$ be some constant to be determined later. Suppose that there exists a $T$-query algorithm where $T < C_1 \left( \sqrt{2^n/2^q} \right)$ which uses a $q$-qubit computational basis state $\ket{y}$, $m = \poly(n)$ workspace qubits initialized in $\ket{0^m}$, and can distinguish between cases (i) and (ii) w.p. $\geq 2/3$. We will show that for some small enough constant $C_1$ this contradicts the known lower bound on computing the OR-function~\cite{beals2001quantum,ambainis2000quantum}, which implies that $T \geq C_1 \left( \sqrt{2^n/2^q} \right)$. W.l.o.g., we can write the output of the verification circuit then as
\begin{align*}
    \ket{\Psi} = U_T O_c U_{T-1} O_c \dots O_c U_1 O_c U_0 \ket{y} \ket{0^{m}},
\end{align*}
where the $O_c$ are controlled queries (tensored with identities) to $\bmO_f$ in either case (i) or case (ii), and the $U_i$ are unitaries which do not depend on $\bmO_f$. Write $\ket{\Psi_{(i)}}$ (resp. $\ket{\Psi_{(ii)}}$) for the final state in case (i) (resp. case (ii). Then in order to have success probability $\geq 2/3$, we require that $\norm{\ketbra{\Psi_{(i)}}-\ketbra{\Psi_{(ii)}} }_1 \geq 1/3$. Hence, we have that $|\bra{\Psi_{(i)}}\ket{\Psi_{(ii)}}|^2 \leq 1-(2/3)^2 = 8/9$. We will now use that in case (ii) the oracle does nothing. Therefore, if we instead consider the circuit which prepares the state $\ket{\Psi}$ followed by all the $U_i^\dagger$ sequentially (without the oracle calls), we have that the final state can be written as
\begin{align*}
    \ket{\Psi'} = \prod_{i =0}^{T} U_i^{\dagger}  U_T O_c U_{T-1} O_c \dots O_c U_1 O_c U_0 \ket{y}\ket{0^m}.
\end{align*}
This means that in case (ii) we have that the final state $\ket{\Psi'_{(ii)}} =  \ket{y} \ket{0^{m}}$. Since the inner product is conserved under unitary transformations, we still have $|\braket{\Psi'_{(i)}}|^2 \leq 8/9$. If we now perform the measurement $M = \{  \ket{y}\ketbra{0^m}{y}\bra{0^m},\mathbb{I}- \ket{y}\ketbra{0^m}{y}\bra{0^m}\}$, we will have that in case (i) we measure $\ket{y}\ket{0^m} $ with probability at most $8/9$ and in case (ii) with probability $1$. If we now instead randomly guess a string $y\in \{0,1\}^q$, we will have in case (ii) that for all corresponding basis states $\ket{y}$ the probability of measuring $\ket{y}\ket{0^m}$ is still $1$ whilst in case (i) there is a probability of $1/2^q$ of sampling a $y$ such that $\ket{y}\ket{0^m}$ is not measured with probability at least $8/9$. By using amplitude amplification, we can boost this back to be $\geq 2/3$ at the cost of a factor of $C_2 \sqrt{2^q}$ more queries for some constant $C_2 >0$~\cite{beals2001quantum}. Hence, this would imply a quantum algorithm for computing the OR-function using $C_1 C_2 \left( \sqrt{2^n} \right) $ queries, which contradicts the lower bound on OR for a small enough constant $C_1$. This implies that $T = \Omega(\sqrt{2^n/2^q})$.
\end{proof}
The key idea behind why this would also imply a lower bound in a PCP setting, is that it allows one to ``fix'' the dishonest prover's strategy in case (ii) whilst still being the optimal witness in case (i). Hence, since this is a valid prover strategy in case (ii), but not necessarily the optimal one, it is a lower bound even in a prover setting.

\orseplog*
\begin{proof}
Let $L$ be a unary language, $p(n)$ be some polynomial and let the oracle $\bmO$ be defined as follows:
\begin{itemize}
    \item if $0^n \in L$, then there is no string $x$ in $\bmO$ of length $ \log^{c+1} n$.
    \item if $0^n \notin L$, then $\bmO$ contains a string $x$ of length $|x| = \log^{c+1} n$.
\end{itemize}
To show that $L \in \QCPCP_Q [\mO(\log^c) n]^{\bmO}$, consider the following protocol: the prover sends a concatenation of $\log^{c} n$ classical proofs $y_i$, each of which encodes the function values of a function $f_{x_i}$ where $x_i$ contains at most a $\lceil \log n \rceil$ part of the total string $x$, i.e.~$x = x_1 x_2 \dots x_k$, where $k\leq \log n$.  The verifier performs a single quantum query on each $y_i$ separately, learning each $x_i$ part of $x$ by using the algorithm from~\cref{lem:BV_string_hiding}, and constructs $x = (x_1,\dots,x_k)$. This requires at most $\mO(\log^c)$ quantum queries. Finally, the verifier makes a single query to the oracle $\bmO$ to check whether $x$ is in $\bmO$, and thus whether $0^n \in L$.

We have that $L \notin \QCPCP [\log^c n]$ follows from the lower bound in~\cref{lem:lb_classical_proof}. By the result of~\cite{weggemans2023guidable}, where it is shown that for $\QCPCP[\mO(\log^c n)]^{\bmO}$ the initial distribution over which the qubits to be accessed can be fixed, even relative to an oracle, we can assume that $\QCPCP[\mO(\log^c n)]^{\bmO}$-verifier is non-adaptive and picks the indices according to a distribution which does not depend on the oracle $\bmO$ (all queries to $\bmO$ will be made after). Write $\bmO_{\emptyset}$ for the case when the oracle is empty and $\bmO_{x}$ when it contains a hidden string $x$. Since the distribution over which the indices are selected does not depend on the oracle, we have
\begin{align*}
    \Pr[V^{\bmO_{\emptyset}} \text{ queries } y_{i_1}, \dots, y_{i_q}] = \Pr[V^{\bmO_{x}} \text{ queries } y_{i_1}, \dots, y_{i_q} ].
\end{align*}
The acceptance probability of the $\QCPCP^{\bmO}[q]$-verifier is given by
\begin{align*}
    \sum_{i_1,\dots, i_q }\Pr[V^{\bmO} \text{ queries } y_{i_1}, \dots, y_{i_q}] \cdot \Pr[V^{\bmO} \text{ accepts querying } y_{i_1}, \dots, y_{i_q}]
\end{align*}
which is simply the expectation value
\begin{align*}
    \mathbb{E}_{i_1,\dots,i_1} \Pr[V^{\bmO} \text{ accepts querying } y_{i_1}, \dots, y_{i_q}].
\end{align*}
Therefore, the maximum difference in accepting given $\bmO_{\emptyset}$ or $\bmO_x$ can be upper bounded by a maximization over the choice of indices as well as the proof, that is
\begin{align*}
    \max_{y}
    \abs{\max_{i_1,\dots,i_1} \left(\Pr[V^{\bmO_\emptyset} \text{ accepts querying } y_{i_1}, \dots, y_{i_q}] - \Pr[V^{\bmO_x} \text{ accepts querying } y_{i_1}, \dots, y_{i_q}] \right)}
\end{align*}
Since~\cref{lem:lb_classical_proof} holds for any configuration of the $q$ qubits that are read, it must also hold the maximization value over any distribution of $q$ qubits (the fact that the choice of indices can change does not matter, as this does not depend on the input and can be absorbed into the first unitary). This means that there exists a string $x$ such that the number of queries needs to be 
\begin{align*}
    \Omega\left( \sqrt{ \frac{2^{\log^{c+1} n}}{2^{\log^{c} n}}} \right),
\end{align*}
which is super-polynomial in $n$. By the standard argument of turning query complexity lower bounds into oracle separations~\cite{baker1975relativizations}, one arrives at the desired result.
\end{proof}
 
\section*{Acknowledgements}
JW thanks Chris Cade, Marten Folkertsma and Alex Grilo for useful discussions. JW was supported by the Dutch Ministry of Economic Affairs and Climate Policy (EZK), as part of the Quantum Delta NL programme. JLG was supported by MEXT Q-LEAP grant No.~JPMXS0120319794 and JSPS KAKENHI grants Nos.~JP20H00579, JP20H05966, and JP	24H00071.

\bibliography{main.bib}

\newcommand{\etalchar}[1]{$^{#1}$}
\begin{thebibliography}{ALM{\etalchar{+}}98}

\bibitem[AALV09]{Aharonov2008The}
Dorit Aharonov, Itai Arad, Zeph Landau, and Umesh Vazirani.
\newblock The detectability lemma and quantum gap amplification.
\newblock In {\em Proceedings of the Forty-First Annual ACM Symposium on Theory
  of Computing (STOC 2009)}, page 417–426, 2009.
\newblock \href{https://arxiv.org/abs/0811.3412}{\tt arXiv:0811.3412}.

\bibitem[AAV13]{aharonov2013guest}
Dorit Aharonov, Itai Arad, and Thomas Vidick.
\newblock Guest column: the quantum {PCP} conjecture.
\newblock {\em ACM SIGACT news}, 44(2):47--79, 2013.
\newblock \href{https://arxiv.org/abs/1309.7495}{\tt arXiv:1309.7495}.

\bibitem[AGKR24]{agarwal2024quantum}
Avantika Agarwal, Sevag Gharibian, Venkata Koppula, and Dorian Rudolph.
\newblock Quantum polynomial hierarchies: {Karp-Lipton}, error reduction, and
  lower bounds, 2024.
\newblock \href{https://arxiv.org/abs/2401.01633}{\tt arXiv:2401.01633}.

\bibitem[AH23]{aaronson2023certified}
Scott Aaronson and Shih-Han Hung.
\newblock Certified randomness from quantum supremacy.
\newblock In {\em Proceedings of the 55th Annual ACM Symposium on Theory of
  Computing (STOC 2023)}, pages 933--944, 2023.
\newblock \href{https://arxiv.org/abs/2303.01625}{\tt arXiv:2303.01625}.

\bibitem[AIK22]{aaronson2021acrobatics}
Scott Aaronson, DeVon Ingram, and William Kretschmer.
\newblock {The Acrobatics of BQP}.
\newblock In {\em Proceedings of the 37th Computational Complexity Conference
  (CCC 2022)}, volume 234 of {\em Leibniz International Proceedings in
  Informatics (LIPIcs)}, pages 20:1--20:17, 2022.
\newblock \href{https://arxiv.org/abs/2111.10409}{\tt arXiv:2111.10409}.

\bibitem[ALM{\etalchar{+}}98]{Arora1998proof}
Sanjeev Arora, Carsten Lund, Rajeev Motwani, Madhu Sudan, and Mario Szegedy.
\newblock Proof verification and the hardness of approximation problems.
\newblock {\em Journal of the ACM}, 45(3):501–555, 1998.

\bibitem[Amb00]{ambainis2000quantum}
Andris Ambainis.
\newblock Quantum lower bounds by quantum arguments.
\newblock In {\em Proceedings of the thirty-second annual ACM symposium on
  Theory of computing (STOC 2000)}, pages 636--643, 2000.
\newblock \href{https://arxiv.org/abs/quant-ph/0002066}{\tt
  arXiv:quant-ph/0002066}.

\bibitem[AS98]{Arora1998probabilistic}
Sanjeev Arora and Shmuel Safra.
\newblock Probabilistic checking of proofs: A new characterization of {NP}.
\newblock {\em Journal of the ACM}, 45(1):70–122, January 1998.

\bibitem[AS24]{arad2024quasi}
Itai Arad and Miklos Santha.
\newblock Quasi-quantum states and the quasi-quantum {PCP} theorem, 2024.
\newblock \href{https://arxiv.org/abs/2410.13549}{\tt arXiv:2410.13549}.

\bibitem[BBC{\etalchar{+}}01]{beals2001quantum}
Robert Beals, Harry Buhrman, Richard Cleve, Michele Mosca, and Ronald de~Wolf.
\newblock Quantum lower bounds by polynomials.
\newblock {\em Journal of the ACM}, 48(4):778--797, 2001.
\newblock \href{https://arxiv.org/abs/quant-ph/9802049}{\tt
  arXiv:quant-ph/9802049}.

\bibitem[BdW02]{buhrman2002complexity}
Harry Buhrman and Ronald de~Wolf.
\newblock Complexity measures and decision tree complexity: a survey.
\newblock {\em Theoretical Computer Science}, 288(1):21--43, 2002.

\bibitem[BGS75]{baker1975relativizations}
Theodore Baker, John Gill, and Robert Solovay.
\newblock Relativizations of the {P=?NP} question.
\newblock {\em SIAM Journal on computing}, 4(4):431--442, 1975.

\bibitem[BHW24]{buhrman2024quantum}
Harry Buhrman, Jonas Helsen, and Jordi Weggemans.
\newblock Quantum {PCPs}: on adaptivity, multiple provers and reductions to
  local hamiltonians, 2024.
\newblock \href{https://arxiv.org/abs/2403.04841}{\tt arXiv:2403.04841}.

\bibitem[BS01]{borchert2001dot}
Bernd Borchert and Riccardo Silvestri.
\newblock Dot operators.
\newblock {\em Theoretical Computer Science}, 262(1):501--523, 2001.

\bibitem[BV93]{bernstein1993quantum}
Ethan Bernstein and Umesh Vazirani.
\newblock Quantum complexity theory.
\newblock In {\em Proceedings of the twenty-fifth annual ACM symposium on
  Theory of computing (STOC 1993)}, pages 11--20, 1993.

\bibitem[DART23]{debris2023quantum}
Thomas Debris-Alazard, Maxime Remaud, and Jean-Pierre Tillich.
\newblock Quantum reduction of finding short code vectors to the decoding
  problem.
\newblock {\em IEEE transactions on Information Theory}, 2023.
\newblock \href{https://arxiv.org/abs/2106.02747}{\tt arXiv:2106.02747}.

\bibitem[Din07]{Dinur2007thepcp}
Irit Dinur.
\newblock The {PCP} theorem by gap amplification.
\newblock {\em Journal of the ACM}, 54(3):12–es, 2007.

\bibitem[Gol06]{goldreich2006promise}
Oded Goldreich.
\newblock On promise problems: A survey.
\newblock In {\em Theoretical Computer Science: Essays in Memory of Shimon
  Even}, pages 254--290. Springer, 2006.

\bibitem[GY19]{gharibian2019complexity}
Sevag Gharibian and Justin Yirka.
\newblock The complexity of simulating local measurements on quantum systems.
\newblock {\em Quantum}, 3:189, 2019.
\newblock \href{https://arxiv.org/abs/1606.05626}{\tt arXiv:1606.05626}.

\bibitem[H\r01]{HastadSome1997}
Johan H\r{a}stad.
\newblock Some optimal inapproximability results.
\newblock {\em Journal of the ACM}, 48(4):798--859, 2001.

\bibitem[JJUW11]{jain2011qip}
Rahul Jain, Zhengfeng Ji, Sarvagya Upadhyay, and John Watrous.
\newblock {QIP}={PSPACE}.
\newblock {\em Journal of the ACM}, 58(6):1--27, 2011.
\newblock \href{https://arxiv.org/abs/0907.4737}{\tt arXiv:0907.4737}.

\bibitem[KST12]{kobler2012graph}
Johannes Kobler, Uwe Sch{\"o}ning, and Jacobo Tor{\'a}n.
\newblock {\em The graph isomorphism problem: its structural complexity}.
\newblock Springer Science \& Business Media, 2012.

\bibitem[KSV02]{Kitaev2002ClassicalAQ}
Alexei~Y. Kitaev, Alexander Shen, and Mikhail~N. Vyalyi.
\newblock {\em Classical and quantum computation}.
\newblock Number~47. American Mathematical Society, 2002.

\bibitem[KW00]{kitaev2000parallelization}
Alexei Kitaev and John Watrous.
\newblock Parallelization, amplification, and exponential time simulation of
  quantum interactive proof systems.
\newblock In {\em Proceedings of the thirty-second annual ACM symposium on
  Theory of computing (STOC 2000)}, pages 608--617, 2000.

\bibitem[Pei09]{Peikert2009Public}
Chris Peikert.
\newblock Public-key cryptosystems from the worst-case shortest vector problem:
  extended abstract.
\newblock In {\em Proceedings of the Forty-First Annual ACM Symposium on Theory
  of Computing (STOC 2009)}, page 333–342, 2009.

\bibitem[Reg09]{regev2009lattices}
Oded Regev.
\newblock On lattices, learning with errors, random linear codes, and
  cryptography.
\newblock {\em Journal of the ACM (JACM)}, 56(6):1--40, 2009.
\newblock \href{https://arxiv.org/abs/2401.03703}{\tt arXiv:2401.03703}.

\bibitem[RR92]{regan1992closure}
Kenneth~W Regan and James~S Royer.
\newblock {\em On closure properties of bounded 2-sided error complexity
  classes}.
\newblock State University of New York at Buffalo, Department of Computer
  Science, 1992.

\bibitem[RT22]{raz2022oracle}
Ran Raz and Avishay Tal.
\newblock Oracle separation of {BQP} and {PH}.
\newblock {\em Journal of the ACM}, 69(4):1--21, 2022.

\bibitem[Sch89]{schoning1989probabilistic}
Uwe Sch{\"o}ning.
\newblock Probabilistic complexity classes and lowness.
\newblock {\em Journal of Computer and System Sciences}, 39(1):84--100, 1989.

\bibitem[Sha92]{shamir1992ip}
Adi Shamir.
\newblock {IP}= {PSPACE}.
\newblock {\em Journal of the ACM}, 39(4):869--877, 1992.

\bibitem[WFC24]{weggemans2023guidable}
Jordi Weggemans, Marten Folkertsma, and Chris Cade.
\newblock {Guidable Local Hamiltonian Problems with Implications to Heuristic
  Ansatz State Preparation and the Quantum PCP Conjecture}.
\newblock In {\em 19th Conference on the Theory of Quantum Computation,
  Communication and Cryptography (TQC 2024)}, volume 310, pages 10:1--10:24,
  2024.
\newblock \href{https://https://arxiv.org/abs/2302.11578}{\tt
  arXiv:2302.11578}.

\bibitem[WJB03]{Wocjan2003two}
Pawel {Wocjan}, Dominik {Janzing}, and Thomas {Beth}.
\newblock {Two QCMA-complete problems}.
\newblock {\em Quantum Information \& Computation}, 3(6):635--643, November
  2003.
\newblock \href{https://arxiv.org/abs/quant-ph/0305090}{\tt
  arXiv:quant-ph/0305090}.

\end{thebibliography}
\bibliographystyle{alpha}

\appendix
\section{Omitted proofs}
\label{app:omittedproofs}

\recursionlemma*
\begin{proof}
We will prove this lemma by induction on $l$.\\

\noindent \textbf{Base Case:} For $l = 0$ we have
\begin{align*}
    (1+\sum_{S \subseteq \emptyset} \prod_{i \in S} a_i)x = (1 + 0)x = x.
\end{align*}

\noindent \textbf{Induction step:} Assume the formula holds for $l = k$, i.e.,
\begin{align*}
    f(k) = (1 + \sum_{S \subseteq [k-1]} \prod_{i \in S} a_i)x.
\end{align*}
Now consider $l = k+1$:
\begin{align*}
    f(k+1) &= x + \sum_{i=0}^{k} a_i f(i) \\
    &= x + \sum_{i=0}^{k} a_i \left(1 + \sum_{S \subseteq [i-1]} \prod_{j \in S} a_j \right)x \\
    &= x\left(1 + \sum_{i=0}^{k} a_i + \sum_{i=0}^{k} a_i \sum_{S \subseteq [i-1]} \prod_{j \in S} a_j \right).
\end{align*}
Notice that the right sum (with the product) inside the parentheses can be interpreted as the sum over
all possible products of the coefficients $a_i$ except for those that only contain $a_i$ itself, which is given by the other sum, effectively generating all possible products of subsets of coefficients up to $k$. This can be represented as
\begin{align*}
    f(k+1) &= x\left(1 + \sum_{S \subseteq [k]} \prod_{i \in S} a_i\right).
\end{align*}
Thus, the formula holds for $l = k+1$.

\end{proof}

\lemboundcoef*
\begin{proof}

Let $y^{S} = (y_1,\dots,y_N)$, $y_i \in \{0,1\}$ for all $i\in [N]$, with $y_i = 1$ if and only if $i \in S$ for some subset $S \subseteq [N]$, $|S| \leq d$. It must hold that for any $S' \subseteq [N]$,
\begin{align*}
    \prod_{i \in S'} y_i^S = 1 \text{ if and only if } S'\subseteq S.
\end{align*}
We also have that
\begin{align*}
     \sum_{S' \subseteq [N], |S'| \leq d} \beta_{S'} \prod_{i \in S'} y^{S}_i \leq \norm{P}_\infty
\end{align*}
and thus
\begin{align*}
     \sum_{S' \subseteq S} \beta_{S'} \leq \norm{P}_\infty,
\end{align*}
keeping only the non-zero terms. Isolating the $S$ term gives
\begin{align*}
      \beta_{S} &\leq \norm{P}_\infty - \sum_{S' \subset S} \beta_{S'} 
\end{align*}
which means
\begin{align}
      \abs{\beta_{S}} &\leq \norm{P}_\infty + \sum_{S' \subset S} \abs{\beta_{S'}}.
      \label{eq:app_up_bound_abs}
\end{align}
Hence, the bound of $|\beta_{S}|$ can be expressed only in terms of the cardinalities of the sets $S'$ and upper bounds on the coefficients of $\abs{\beta_{S'}}$. Since all terms are nonnegative, we have that~\cref{eq:app_up_bound_abs} can itself be upper bounded through a recursion
\begin{align*}
    B(l) = \norm{P}_\infty +  \sum_{i=0}^{d-1} \binom{d}{i} B(i)
\end{align*}
with $B(1) = \norm{P}_\infty$. According to~\cref{lem:recursion}, we have that
\begin{align*}
|\beta_S| \leq \left(1+ \sum_{S' \subseteq [d-1]} \prod_{i \in S'} a_i \right) \norm{P}_\infty
\end{align*}
where 
\begin{align*}
a_i = \binom{d}{i}.
\end{align*}
Since for any $i$ it must hold that
\begin{align*}
\binom{d}{i} \leq 2^d,
\end{align*}
we can upper bound $|\beta_S| $ as
\begin{align*}
\left(1+ \sum_{S \subseteq [d-1]} \prod_{i \in S} 2^d \right) \norm{P}_\infty.
\end{align*}
The sum can be evaluated as
\begin{align*}
\sum_{S \subseteq [d-1]} \prod_{i \in S} 2^d= \sum_{k=0}^{d-1} \binom{d-1}{k} 2^{dk} = (1+2^d)^{d-1},
\end{align*}
evoking the binomial theorem. Putting everything together, we find an upper bound of
\begin{align*}
|\beta_S| \leq \left(1 + (1 + 2^d)^{d-1}\right) \norm{P}_\infty,
\end{align*}
as was to be shown.

\end{proof}

\end{document}